\documentclass[preprint,12pt]{elsarticle}

\usepackage{a4wide,url}
\usepackage{todonotes}
\usepackage{amsmath}
\usepackage{amssymb}
\usepackage{amsfonts}

\usepackage{amsthm}
\usepackage{hyperref}
\usepackage{stmaryrd}

\newcommand{\ta}{\ensuremath{\mathtt{a}}}
\newcommand{\tb}{\ensuremath{\mathtt{b}}}
\newcommand{\tc}{\ensuremath{\mathtt{c}}}
\newcommand{\td}{\ensuremath{\mathtt{d}}}

\DeclareMathOperator{\eword}{\varepsilon}

\newcommand{\sign}{\textsf{sign}}

\newcommand{\powerset}[1]{\mathcal{P}(#1)}

\newcommand{\colmerge}{\curlyvee}
\newcommand{\colrename}{\rho}
\newcommand{\reg}{\mathsf{REG}}
\newcommand{\cfl}{\mathsf{CFL}}
\newcommand{\mso}{\mathsf{MSO}}

\newcommand{\slice}[1]{\mathsf{sl}_{#1}}

\newcommand{\encoding}{\mathsf{enc}}

\theoremstyle{plain}
\newtheorem{theorem}{Theorem}[section]
\newtheorem{proposition}[theorem]{Proposition}
\newtheorem{lemma}[theorem]{Lemma}
\newtheorem{definition}[theorem]{Definition}
\newtheorem{example}[theorem]{Example}

\theoremstyle{remark}
\newtheorem{observation}{Observation}

\begin{document}

\begin{frontmatter}

\title{A General Information Extraction Framework\\ Based on Formal Languages}

\author[aff]{Markus L.\ Schmid} 
\ead{MLSchmid@MLSchmid.de}

\affiliation[label1]{organization={Humboldt University Berlin},
             addressline={Unter den Linden 6},
             city={Berlin},
             postcode={D-10099},
             state={Berlin},
             country={Germany}}

\begin{abstract}
For a terminal alphabet $\Sigma$ and an attribute alphabet $\Gamma$, a $(\Sigma, \Gamma)$-extractor is a function that maps every string $w$ over $\Sigma$ to a table with a column per attribute and with sets of positions of $w$ as cell entries. This rather general information extraction framework extends the well-known document spanner framework, which has intensively been investigated in the database theory community over the last decade. Moreover, our framework is based on formal language theory in a particularly clean and simple way. In addition to this conceptual contribution, we investigate closure properties, different representation formalisms and the complexity of natural decision problems for extractors.
\end{abstract}

\begin{keyword}
Information Extraction \sep Regular Languages \sep Context-Free Languages \sep Database Theory



\end{keyword}

\end{frontmatter}

\section{Introduction}

Over roughly the last decade, the data query paradigm of so-called \emph{information extraction} has received a lot of attention in database theory. In a nutshell, information extraction is the task to extract from a text document (string, word, sequence, etc.) a relational table of structured data. The most famous instance of information extraction are so-called \emph{document spanners} (introduced in~\cite{FaginEtAl2015}), which are based on \emph{spans}, e.\,g., $(3, 6)$ is a span of $w = \ta \tb \ta \ta \tb \tc \tb \ta$ referring to the factor $w[3..6] = \ta \ta \tb \tc$. A document spanner has a fixed set of variables and from any string $w$ it extracts a table with a column per variable and with spans of $w$ as cell entries (every row of the table is considered to be a result tuple of the spanner). 
We refer to the papers~\cite{DBLP:journals/sigmod/AmarilliBMN20,Schmid2024,SchmidSchweikardt2022} for general information about document spanners, and~\cite{FreydenbergerEtAl2018,PeterfreundEtAl2019,AmarilliEtAl2021,MaturanaEtAl2018,FlorenzanoEtAl2020,PeterfreundEtAl2019_2,Freydenberger2019,FreydenbergerHolldack2018,FreydenbergerThompson2020,FreydenbergerThompson2022,Peterfreund2019PhD,Peterfreund2023,DoleschalEtAl2020,SchmidSchweikardt2024,SchmidSchweikardtPODS2021,SchmidSchweikardtPODS2022,MunozRiveros2025,DoleschalEtAl2023,AmarilliEtAl2022,GawrychowskiEtAl2024} for recent results. Document spanners are attractive from a formal languages point of view, since representations and algorithmic techniques are based on classical concepts from formal languages and automata theory. 

We propose a more general information extraction framework, which properly extends that of document spanners.
Our motivation is as follows:

\begin{enumerate}
\item Several deep algorithmic techniques developed for document spanners also work without additional effort in a more general setting (e.\,g.,~\cite{AmarilliEtAl2021,AmarilliEtAl2022,GawrychowskiEtAl2024,MunozRiveros2025}), so it makes sense to formally define and investigate this setting. 
\item Our general framework embeds into classical formal language theory in an even cleaner way, i.e., the equivalence of classes of extractors and classes of formal languages is more explicit. Hence, our framework may serve as an interface especially tailored to formal language and automata theorists. 
\item Our information extraction framework seems to occupy an interesting area between language descriptors and transducers; this is worth investigating in pure terms of formal language theory (in contrast to the work in database theory, which focuses on solving data management tasks).
\end{enumerate}

\subsection{Intuitive Explanation}\label{sec:intuition}

We consider strings over a finite \emph{terminal alphabet} $\Sigma$ as data objects that we want to query. Our queries -- called \emph{extractors} -- will extract a table whose cells contain sets of positions of the terminal string:

\begin{center}
$w = \ta \tb \ta \tc \ta \td \tc \td \td$ \hspace{1cm} $\implies$ \hspace{1cm}
{\scriptsize
\begin{tabular}{|l|l|l|}\hline
$x$ & $y$ & $z$\\\hline
$\{4, 7, 9\}$ & $\{8, 9\}$ & $\{1, 3\}$\\\hline
$\{4, 7\}$ & $\{6\}$ & $\{1, 5\}$\\\hline
$\emptyset$ & $\{2, 5\}$ & $\{5\}$\\\hline
$\{7\}$ & $\emptyset$ & $\{7\}$\\\hline
\end{tabular}
}
\end{center}

Here, $\Gamma = \{x, y, z\}$ is the set of attributes that label the columns. The table can contain the empty set as entries (as illustrated above) and it can also be completely empty (i.e., it has no rows). 

This extends the setting of document spanners by replacing spans (i.e., exactly two positions of the string) by arbitrary sets of positions (thus, the setting obviously still covers document spanners).

Moreover, such extractor functions can be easily seen as formal languages. For the purpose of illustration, let $E$ be the extractor function from above, i.e., $E(w)$ is the table presented above for $w = \ta \tb \ta \tc \ta \td \tc \td \td$. Then we can represent $E$ by a formal language $L_E$ that, for every string $w$ over $\Sigma$ and every row of the table $E(w)$, contains a string representation of that row. Such a string representation is just a version of $w$ where the extracted information is represented by annotations. For example, the first row $(\{4, 7, 9\}, \{8, 9\}, \{1, 3\})$ is represented by the string 
\begin{equation*}
\{z\}_\ta \: \emptyset_\tb \: \{z\}_\ta \: \{x\}_\tc \: \emptyset_\ta \: \emptyset_\td \: \{x\}_\tc \: \{y\}_\td \: \{x, y\}_\td 
\end{equation*}
over the alphabet of all symbols of the form $X_\sigma$ with $X \subseteq \Gamma$ and $\sigma \in \Sigma$ (such symbols are called \emph{markers} and strings of markers are called \emph{marker strings}). Obviously, a marker $X_\sigma$ at position $i$ simply means that the corresponding row is such that $i$ is contained in all the sets of the columns of $X$. Hence, the above table $E(w)$ gives rise to four such marker strings of the language $L_E$, and the complete language $L_E$ results from collecting all marker strings for every possible table that is an $E$-image of some input string. 

This also means that any marker language (i.e., a set of marker strings) $L$ describes a well-defined extractor function $\llbracket L \rrbracket$. For example, if $L$ contains a marker string 
\begin{equation*}
\emptyset_\ta \: \{x, y, z\}_\tb \: \emptyset_\tb \: \{x\}_\ta \: \{z\}_\tc \: \emptyset_\ta\,, 
\end{equation*}
then this simply means that $(\{2, 4\}, \{2\}, \{2, 5\})$ is a row of the table $\llbracket L \rrbracket(\ta \tb \tb \ta \tc \ta)$. The difference to the framework of document spanners is that in order to describe classes of extractors, we do not have to do that on an operational level (i.e., by designing a machine model that characterises the extractor functions of that class), but we can stay on a purely declarative level (i.e., giving a class of marker languages is sufficient). 

No matter whether we consider a regular expression $(\emptyset_{\ta} \vee \emptyset_{\tb})^* (\{x\}_{\tb} \vee \{y\}_{\tc})^+ (\{z\}_{\ta} \vee \{z\}_{\tb})^*$, or a nondeterministic finite automaton like the one of Figure~\ref{fig:NFA}, or a context-free grammar $S \to \emptyset_{\ta} \: S \: \emptyset_{\ta} \mid \{x\}_{\tb} \mid \{y\}_{\tc}$, or literally any other descriptor of a marker language, this always represents a well-defined extractor.

\begin{figure}
\begin{center}
\includegraphics[scale=0.6]{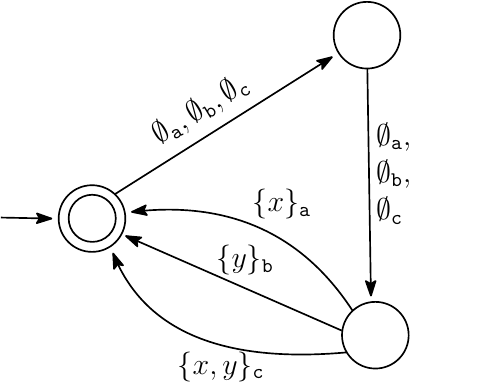}
\end{center}
\caption{An NFA that accepts a marker language and therefore represents an extractor with terminal alphabet $\{\ta, \tb, \tc\}$ and attributes $\{x, y, z\}$.}
\label{fig:NFA}
\end{figure}

\subsection{Contributions of this Work} 

Our main conceptual contribution is the introduction of the general information extraction framework (Section~\ref{sec:formalFramework}). We define several operators on extractors in Section~\ref{sec:operations}. In Section~\ref{sec:extractorsFormalLanguages}, we show that our extractors have a convenient formulation as formal languages (as sketched in Section~\ref{sec:intuition}), and their operators translate into natural language operations. In the rest of the paper, we focus on classes of extractors that can be described by finite automata and context-free grammars. We investigate closure properties, different representation formalisms (Section~\ref{sec:regCFExtractors}), and the complexity of several natural decision problems (Section~\ref{sec:compProblems}).

\subsection{Basic Definitions}

For an alphabet $A$, we denote by $A^*$ the set of words over $A$, and $A^+ = A^* \setminus \{\eword\}$, where the symbol $\eword$ is used to denote the empty word. Let $\reg_A$ and $\cfl_A$ denote the classes of regular and context-free languages over alphabet $A$, respectively. We use nondeterministic finite automata (NFA), deterministic finite automata (DFA) and context-free grammars (CFG) as commonly defined (see, e.\,g.,~\cite{HopcroftEtAlBook2007}). By $\powerset{A}$ we denote the power set of a set $A$. For a string $w$, we use $w[i]$ for $i \in \{1, 2, \ldots, |w|\}$ to denote the $i^{\text{th}}$ letter of $w$, and $w[i..j]$ for $i, j \in \{1, 2, \ldots, |w|\}$ with $i \leq j$ to denote the factor $w[i] w[i+1] \ldots w[j]$. We will generally use the symbol $\bot$ for signifying ``undefined''.

\section{Formal Definition of the Framework}\label{sec:formalFramework}

Our extractor framework is particularly simple and requires only a few basic definitions to be presented next. We shall discuss a brief example at the end of this section.

Let $\Sigma$ be a finite \emph{terminal alphabet} and let $\Gamma$ be a finite \emph{attribute alphabet}, and every $x \in \Gamma$ is called an \emph{attribute symbol}. For complexity considerations, we let $\Sigma$ be constant. The set $\Gamma$ will play the role of attributes of the extracted tables as explained above, i.e., the columns will be labelled by the attribute symbols from $\Gamma$. Formally, we represent the row of a table extracted from $w \in \Sigma^*$ as a \emph{$\Gamma$-tuple} (\emph{for $w$}), which is a function $t : \Gamma \to \powerset{\{1, 2, \ldots, |w|\}}$. For convenience, we also call $t(x)$ the \emph{$x$-entry} of $t$, and for the sake of presentation, we sometimes assume a fixed order $\preceq$ on $\Gamma$ and then represent $t$ in tuple notation, i.e., as $(t(x), t(y), t(z))$, where $x \preceq y \preceq z$. Moreover, we denote by $t^{\emptyset}_{\Gamma}$ the \emph{empty $\Gamma$-tuple}, i.e., $t^{\emptyset}_{\Gamma}(x) = \emptyset$ for every $x \in \Gamma$; note that $t^{\emptyset}_{\Gamma}$ is the only $\Gamma$-tuple for the empty word $\eword$. A \emph{$\Gamma$-table} (\emph{for $w$}) is any (possibly empty) set of $\Gamma$-tuples for $w$, and a \emph{$(\Sigma, \Gamma)$-extractor} is a total function $E$ that maps every terminal string $w \in \Sigma^*$ to a $\Gamma$-table for $w$. For convenience, we represent a $\Gamma$-table $T$ by listing its $\Gamma$-tuples (in any order) in tuple notation. The \emph{support} of a $(\Sigma, \Gamma)$-extractor $E$ is $\{w \in \Sigma^* \mid E(w) \neq \emptyset\}$, and we say that $E$ has \emph{finite support} if the support is a finite set. The \emph{empty $\Gamma$-extractor} $E_{\Gamma}^{\emptyset}$ is defined by $E_{\Gamma}^{\emptyset}(w) = \emptyset$ for every $w \in \Sigma^+$ and $E_{\Gamma}^{\emptyset}(\eword) = \{t^{\emptyset}_{\Gamma}\}$; note that $E_{\Gamma}^{\emptyset}$ is different from the $\Gamma$-extractor that maps every string (including $\eword$) to $\emptyset$.

\begin{figure}[!b]
\begin{center}
{\scriptsize
\begin{tabular}{|l|l|l|}\hline
\multicolumn{3}{|l|}{$E_1(w)$} \\\hline\hline
$x$ & $y$ & $z$ \\\hline
$\{2, 8\}$ & $\{2, 3\}$ & $\{7, 9, 10\}$ \\\hline
$\{2, 8\}$ & $\{2, 3\}$ & $\{9, 10\}$ \\\hline
$\{2, 8\}$ & $\{2, 3\}$ & $\{7, 10\}$ \\\hline
$\{2, 8\}$ & $\{2, 3\}$ & $\{7, 9\}$ \\\hline
$\{2, 8\}$ & $\{2, 3\}$ & $\{10\}$ \\\hline
$\{2, 8\}$ & $\{2, 3\}$ & $\{9\}$ \\\hline
$\{2, 8\}$ & $\{2, 3\}$ & $\{7\}$ \\\hline
$\{2, 8\}$ & $\{2, 3\}$ & $\emptyset$ \\\hline
\end{tabular}
\hspace{1cm}
\begin{tabular}{|l|l|}\hline
\multicolumn{2}{|l|}{$E_2(w)$} \\\hline\hline
$A$ & $B$ \\\hline
$\{1, 7\}$ & $\{7, 10\}$  \\\hline
$\{1, 10\}$ & $\{7, 10\}$  \\\hline
$\{5, 7\}$ & $\{7, 10\}$  \\\hline
$\{5, 10\}$ & $\{7, 10\}$  \\\hline
$\{1, 7\}$ & $\{7\}$  \\\hline
$\{1, 10\}$ & $\{7\}$ \\\hline 
$\{5, 7\}$ & $\{7\}$ \\\hline
$\{5, 10\}$ & $\{7\}$  \\\hline
$\{1, 7\}$ & $\{10\}$\\\hline
$\{1, 10\}$ & $\{10\}$\\\hline
$\{5, 7\}$ & $\{10\}$\\\hline
$\{5, 10\}$ & $\{10\}$\\\hline
$\{1, 7\}$ & $\emptyset$\\\hline
$\{1, 10\}$ & $\emptyset$\\\hline
$\{5, 7\}$ & $\emptyset$\\\hline
$\{5, 10\}$ & $\emptyset$\\\hline
\end{tabular}
}

\end{center}
\caption{The tables extracted from $w = \tb \ta \ta \ta \tb \ta \tc \ta \td \tc \tb$ by $E_1$ and $E_2$ from Example~\ref{ExtractorExample}.}
\label{fig:extractorExample}
\end{figure}

Note that every $\Gamma$-tuple for $w$ and every $\Gamma$-table for $w$ depend on $\Gamma$ and $|w|$, but not on the actual content of $w$, i.e., not on the terminal alphabet $\Sigma$. On the other hand, $(\Sigma, \Gamma)$-extractors depend on $\Gamma$ \emph{and} $\Sigma$, because $\Sigma^*$ is their domain. 

We shall also talk about tuples and tables instead of $\Gamma$-tuples and $\Gamma$-tables whenever $\Gamma$ is clear from the context or not important. With respect to extractors, we will also talk about $\Sigma$-extractors or $\Gamma$-extractors whenever $\Gamma$ or $\Sigma$, respectively, is clear from the context or negligible, or just about extractors when both alphabets are clear from the context or negligible. 
The term $(\Sigma, \Gamma_{\downarrow})$-extractor refers to an arbitrary $(\Sigma, \Gamma')$-extractor with $\Gamma' \subseteq \Gamma$.

\begin{example}\label{ExtractorExample}
Let $\Sigma = \{\ta, \tb, \tc, \td\}$, let $\Gamma_1 = \{x, y, z\}$ and let $\Gamma_2 = \{A, B\}$. Let $E_1$ be the $(\Sigma, \Gamma_1)$-extractor that maps every $w \in \Sigma^*$ to the set of $\Gamma_1$-tuples with an $x$-entry $\{i, j\}$, where $i$ is the first and $j$ the last occurrence of $\ta$ in $w$ (or $\emptyset$ if $w$ does not contain any occurrences of $\ta$), a $y$-entry that contains all starting positions of factors $\ta^\ell$ with $\ell \geq 2$, and a $z$-entry that is a subset of the occurrences of $\tc$ and $\td$. Let $E_2$ be the $(\Sigma, \Gamma_2)$-extractor that maps every $w \in \Sigma^*$ to the set of $\Gamma_2$-tuples with an $A$-entry $\{i, j\}$, where $i < j$, $w[i] = b$ and $w[j] = c$, and a $B$-entry that contains occurrences of $c$ (i.e., it is a subset of the set of all occurrences of $\tc$). For the string
\begin{align*}
&w =& &\tb& &\ta& &\ta& &\ta& &\tb& &\ta& &\tc& &\ta& &\td& &\tc& &\tb&\\
& & &1& &2& &3& &4& &5& &6& &7& &8& &9& &10& &11&
\end{align*}
the $\Gamma_1$-table $E_1(w)$ and the $\Gamma_2$-table $E_2(w)$ is shown in Figure~\ref{fig:extractorExample}.
\end{example}

We summarise that for a fixed terminal alphabet $\Sigma$ and a fixed attribute alphabet $\Gamma$, we have defined the class of $(\Sigma, \Gamma)$-extractors. Moreover, as discussed in the introduction, such extractors are a formalisation of information extraction queries as investigated in database theory.

\section{Operations on Extractors}\label{sec:operations}

In this section, we define several operators on extractors based on typical operators for manipulating relational tables (i.e., operators from relational algebra), but we also consider the concatenation and Kleene star, which is motivated by our formulation of the framework purely in terms of formal languages (to be discussed in Section~\ref{sec:extractorsFormalLanguages}). Some of our binary operations require that the two extractors are defined over the same attribute alphabets, others are defined for arbitrary attribute alphabets. Note that in this section, we assume that $\Sigma$ is a fixed terminal alphabet and all our considered extractors are $\Sigma$-extractors. However, we consider extractors over different attribute alphabets.

\subsection{Set Operations} 

If $E_1$ and $E_2$ are $\Gamma$-extractors and $\odot \in \{\cup, \cap, \setminus\}$, then $E_1 \odot E_2$ is the $\Gamma$-extractor defined by $(E_1 \odot E_2)(w) = E_1(w) \odot E_2(w)$ for every $w \in \Sigma^*$, and $\neg E_1$ is a $\Gamma$-extractor defined by $\neg E_1(w) = \overline{E_1(w)}$ for every $w \in \Sigma^*$. Note that for the complement $\overline{T_1}$ and the set difference $T_1 \setminus T_2$ for $\Gamma$-tables $T_1, T_2$ for $w$, we consider the set of all $\Gamma$-tuples for $w$ as the universe.

\subsection{Join Variants}\label{sec:joinvariants}

Let $E_1$ and $E_2$ be $\Gamma_1$- and $\Gamma_2$-extractors, respectively, let $T_1$ and $T_2$ be $\Gamma_1$- and $\Gamma_2$-tables for $w_1$ and $w_2$, respectively, and let $t_1$ and $t_2$ be $\Gamma_1$- and $\Gamma_2$-tuples for $w_1$ and $w_2$, respectively. 

The product $\times$ and the (natural) join operation $\bowtie$ are well-known operators on tables, so they can be lifted to operators on extractors, i.e., for every $\circ \in \{\times, \bowtie\}$, we have $(E_1 \circ E_2)(w) = E_1(w) \circ E_2(w)$ for every $w \in \Sigma^*$ (note that $E_1 \times E_2$ is only defined if $\Gamma_1 \cap \Gamma_2 = \emptyset$).\footnote{If $\Gamma_1 \cap \Gamma_2 = \emptyset$, then product and join are the same, so we shall not consider the product in the rest of the paper.}

Recall that the join $t_1 \bowtie t_2$ is only defined if the tuples agree on their common attributes from $\Gamma_1 \cap \Gamma_2$, i.e., $t_1(x) = t_2(x)$ for every $x \in \Gamma_1 \cap \Gamma_2$. However, in our setting of information extraction, the entries of tuples are all subsets of $\mathbb{N}$. Thus, instead of letting $t_1 \bowtie t_2$ be undefined if there is some $x \in \Gamma_1 \cap \Gamma_2$ with $t_1(x) \neq t_2(x)$, we could also define $(t_1 \bowtie t_2)(x) = t_1(x) \odot t_2(x)$ for some $\odot \in \{\cup, \cap, \setminus\}$ in this case, which motivates the following generalised join variants. 

For $\odot \in \{\cup, \cap, \setminus\}$, the \emph{$\odot$-join} $t_1 \bowtie_{\odot} t_2$ is a $(\Gamma_1 \cup \Gamma_2)$-tuple defined by 

\begin{equation*}
(t_1 \bowtie_{\odot} t_2)(x) = 
\begin{cases}
t_1(x)& \text{if $x \in \Gamma_1 \setminus \Gamma_2$},\\
t_2(x)& \text{if $x \in \Gamma_2 \setminus \Gamma_1$},\\
t_1(x) \odot t_2(x)& \text{if $x \in \Gamma_1 \cap \Gamma_2$}.
\end{cases}
\end{equation*} 

Observe that the normal join is therefore given by $t_1 \bowtie t_2 = \bot$ if $t_1(x) \neq t_2(x)$ for some $x \in \Gamma_1 \cap \Gamma_2$, and $t_1 \bowtie t_2 = t_1 \bowtie_{\cup} t_2$ otherwise. 

Each $\circ \in \{\bowtie_{\cup}, \bowtie_{\cap}, \bowtie_{\setminus}, \bowtie)$ extends to tables and extractors in the obvious way, i.e., the $(\Gamma_1 \cup \Gamma_2)$-table $T_1 \circ T_2$ is defined by $T_1 \circ T_2 = \{t_1 \circ t_2 \mid t_1 \in T_1, t_2 \in T_2\}$, and the $(\Gamma_1 \cup \Gamma_2)$-extractor $E_1 \circ E_2$ is defined by $(E_1 \circ E_2)(w) = E_1(w) \circ E_2(w)$ for every $w \in \Sigma$. 

Let us discuss some examples of these join operations on tables. The following is an example for the normal join $\bowtie$:

{\scriptsize
\begin{center}
\begin{tabular}{|l|l|l|}\hline
$u$ & $v$ & $w$\\\hline
$\{1, 4\}$ & $\{2, 3\}$ & $\{7\}$\\\hline
$\emptyset$ & $\{4, 6, 7\}$ & $\{1\}$\\\hline
$\{5\}$ & $\{2\}$ & $\{4, 5\}$\\\hline
\end{tabular}
$\bowtie$
\begin{tabular}{|l|l|l|}\hline
$u$ & $v$ & $x$\\\hline
$\{5\}$ & $\{2\}$ & $\emptyset$\\\hline
$\emptyset$ & $\{4, 6, 7\}$ & $\{4, 5\}$\\\hline
$\{5\}$ & $\{2\}$ & $\{1, 2, 7\}$\\\hline
$\{1, 4\}$ & $\{1, 2, 3\}$ & $\{7\}$\\\hline
\end{tabular}
$=$
\begin{tabular}{|l|l|l|l|}\hline
$u$ & $v$ & $w$ & $x$\\\hline
$\emptyset$ & $\{4, 6, 7\}$ & $\{1\}$ & \{4, 5\}\\\hline
$\{5\}$ & $\{2\}$ & $\{4, 5\}$ & $\emptyset$\\\hline
$\{5\}$ & $\{2\}$ & $\{4, 5\}$ & $\{1, 2, 7\}$\\\hline
\end{tabular}
\end{center}
}

As an example for the generalised join operations, let $T_1$ be an $\{x, y\}$-table defined by $T_1 = ((\{1, 2\}, \{4\}), (\emptyset, \{1\}))$ and $T_2$ an $\{x, z\}$-table with $T_2 = ((\{2, 3\}, \{3, 7\}),(\{4\}, \{1, 3\}))$ (we assume that $x \preceq y \preceq z$). Then we have:

{\scriptsize
\begin{center}
\begin{tabular}{|l|l|l|}\hline
\multicolumn{3}{|l|}{$T_1 \bowtie_{\cup} T_2$} \\\hline\hline
$x$ & $y$ & $z$\\\hline
$\{1, 2, 3\}$ & $\{4\}$ & $\{3, 7\}$\\\hline
$\{1, 2, 4\}$ & $\{4\}$ & $\{1, 3\}$\\\hline
$\{2, 3\}$ & $\{1\}$ & $\{3, 7\}$\\\hline
$\{4\}$ & $\{1\}$ & $\{1, 3\}$\\\hline
\end{tabular}
\hspace{0.5cm}
\begin{tabular}{|l|l|l|}\hline
\multicolumn{3}{|l|}{$T_1 \bowtie_{\cap} T_2$} \\\hline\hline
$x$ & $y$ & $z$\\\hline
$\{2\}$ & $\{4\}$ & $\{3, 7\}$\\\hline
$\emptyset$ & $\{4\}$ & $\{1, 3\}$\\\hline
$\emptyset$ & $\{1\}$ & $\{3, 7\}$\\\hline
$\emptyset$ & $\{1\}$ & $\{1, 3\}$\\\hline
\end{tabular}
\hspace{0.5cm}
\begin{tabular}{|l|l|l|}\hline
\multicolumn{3}{|l|}{$T_1 \bowtie_{\setminus} T_2$} \\\hline\hline
$x$ & $y$ & $z$\\\hline
$\{1\}$ & $\{4\}$ & $\{3, 7\}$\\\hline
$\{1, 2\}$ & $\{4\}$ & $\{1, 3\}$\\\hline
$\emptyset$ & $\{1\}$ & $\{3, 7\}$\\\hline
$\emptyset$ & $\{1\}$ & $\{1, 3\}$\\\hline
\end{tabular}
\end{center}
}

\subsection{Concatenation and Kleene-Star}\label{sec:concatKleene}

Again, let $E_1$ and $E_2$ be $\Gamma_1$- and $\Gamma_2$-extractors, respectively, let $T_1$ and $T_2$ be $\Gamma_1$- and $\Gamma_2$-tables for $w_1$ and $w_2$, respectively, and let $t_1$ and $t_2$ be $\Gamma_1$- and $\Gamma_2$-tuples for $w_1$ and $w_2$, respectively. The $(\Gamma_1 \cup \Gamma_2)$-tuple $t_1 \cdot t_2$ for $w_1 \cdot w_2$ is defined by

\begin{equation*}
(t_1 \cdot t_2)(x) = 
\begin{cases}
t_1(x)& \text{if $x \in \Gamma_1 \setminus \Gamma_2$},\\
\{i + |w_1| \mid i \in t_2(x)\}& \text{if $x \in \Gamma_2 \setminus \Gamma_1$},\\
t_1(x) \cup \{i + |w_1| \mid i \in t_2(x)\}& \text{if $x \in \Gamma_1 \cap \Gamma_2$}.
\end{cases}
\end{equation*}

As an example, let $\Gamma_1 = \{u, v, w\}$ and $\Gamma_2 = \{u, v, x\}$ (with assumed order $u \preceq v \preceq w \preceq x$ for the tuple notation), let $w_1 = \ta \tb \ta$ and $w_2 = \tb \ta \tc \td$, and let $t_1 = (\{1\}, \{1, 2\}, \emptyset)$ be a $\Gamma_1$- and $t_2 = (\{1, 3\}, \{2\}, \{4\})$ a $\Gamma_2$-tuple for $w_1$ and $w_2$, respectively. Then, intuitively speaking, the concatenation operation first shifts all values of $t_2$ by $|w_1| = 3$, i.e., we get an intermediate $\Gamma_2$-tuple $t'_2 = (\{4, 6\}, \{5\}, \{7\})$, which is then combined with $t_1$ in the same way as the $\cup$-join operation works, i.e., $t_1 \cdot t_2 = t_1 \bowtie_{\cup} t'_2 = (\{1, 4, 6\}, \{1, 2, 5\}, \emptyset, \{7\})$. Also note that if $t_1 = t^{\emptyset}_{\Gamma_1}$ and $w_1 = \eword$, then $t_1 \cdot t_2 = t_2$.

The $(\Gamma_1 \cup \Gamma_2)$-table $T_1 \cdot T_2$ for $w_1 \cdot w_2$ is defined by $T_1 \cdot T_2 = \{t_1 \cdot t_2 \mid t_1 \in T_1, t_2 \in T_2\}$, and the $(\Gamma_1 \cup \Gamma_2)$-extractor $E_1 \cdot E_2$ is defined by $(E_1 \cdot E_2)(w) = \bigcup_{w = w_1 \cdot w_2} E_1(w_1) \cdot E_2(w_2)$ for every $w \in \Sigma^*$. Note that the empty $\Gamma$-extractor $E^{\emptyset}_{\Gamma}$ satisfies $E^{\emptyset}_{\Gamma} \cdot E = E$ for any $\Gamma$-extractor $E$.

We next observe that the concatenation on tuples, tables and extractors is associative, which can be easily shown (a full proof is provided in~\ref{sec:operationsAppendix}).

\begin{proposition}\label{associativityProp}
For $i \in \{1, 2, 3\}$, let $t_i$ be a $\Gamma_i$-tuple for $w_i$, let $T_i$ be a $\Gamma_i$-table for $w_i$, and let $E_i$ be a $\Gamma_i$-extractor. Then we have $(t_1 \cdot t_2) \cdot t_3 = t_1 \cdot (t_2 \cdot t_3)$, $(T_1 \cdot T_2) \cdot T_3 = T_1 \cdot (T_2 \cdot T_3)$ and $(E_1 \cdot E_2) \cdot E_3 = E_1 \cdot (E_2 \cdot E_3)$.
\end{proposition}

Due to this associativity, we can lift the concatenation to a Kleene-star operator in the usual way. For a $\Gamma$-extractor $E$, we define $(E)^0 = E^{\emptyset}_{\Gamma}$ and $(E)^k = (E)^{k-1} \cdot E$ for every $k \geq 2$. Finally, we set $E^* = \bigcup_{k \geq 1} (E)^k$. Note that we always have $(E^*)(\eword) = \{t^{\emptyset}_{\Gamma}\}$ for every $\Gamma$-extractor $E$, and if $E$ is a $\Gamma$-extractor, then $E^*$ is also a $\Gamma$-extractor.

\subsection{Other Unary Operators}\label{sec:otherOperators}

There are several other natural unary operators that could be defined. We define and discuss some of those that are natural.

Let $\Gamma$ be an attribute alphabet and let $\Gamma' \subseteq \Gamma$. The \emph{$\Gamma'$-projection} of a $\Gamma$-tuple $t$ is the $\Gamma'$-tuple $\pi_{\Gamma'}(t)$ obtained from $t$ by restricting its domain to $\Gamma'$, i.e., $\pi_{\Gamma'}(t)(x) = t(x)$ for every $x \in \Gamma'$. Related to the projection is the \emph{merge} operation. Let $t$ be a $\Gamma$-tuple, let $x, y \in \Gamma$ and let $\odot \in \{\cup, \cap, \setminus\}$. Then the \emph{$\odot$-merge} $\colmerge_{x, y, \odot}(t)$ of $t$ is a $\Gamma \setminus \{y\}$-tuple defined by $\colmerge_{x, y, \odot}(t)(z) = t(z)$ for every $z \in \Gamma \setminus \{x, y\}$ and $\colmerge_{x, y, \odot}(t)(x) = t(x) \odot t(y)$. The \emph{attribute renaming} is an operation that renames a column of a tuple. Let $t$ be a $\Gamma$-tuple, let $z \in \Gamma$ and let $y \notin \Gamma$, then $\colrename_{z \to y}(t)$ is a $((\Gamma \setminus \{z\}) \cup \{y\})$-tuple defined by $(\colrename_{z \to y}(t))(x) = t(x)$ for every $x \in \Gamma \setminus \{z\}$ and $(\colrename_{z \to y}(t))(y) = t(z)$. 

We lift these operations to tables and to extractors in the obvious way. More precisely, for a $\Gamma$-table $T$, a $\Gamma$-extractor $E$ and for every $f \in \{\pi_{\Gamma'}, \colmerge_{y, y', \cup}$, $\colmerge_{y, y', \cap}$, $\colmerge_{y, y', \setminus}$, $\colrename_{y \to z}\}$, we set $f(T) = \{f(t) \mid t \in T\}$ and $f(E)(w) = f(E(w))$ for every $w \in \Sigma^*$.

Intuitively speaking, the $\Gamma'$-projection $\pi_{\Gamma'}(\cdot)$ can be interpreted as removing all columns from the tables that are labelled by an attribute from $\Gamma \setminus \Gamma'$, the $\odot$-merge $\colmerge_{x, y, \odot}(t)$ can be seen as a $(\Gamma \setminus \{y\})$-projection, but the removed column with attribute $y$ is merged with column $x$ via the set operation $\odot$, and the attribute renaming $\colrename_{z \to y}$ simply relabels column $z$ to $y$.

\section{Extractors as Formal Languages}\label{sec:extractorsFormalLanguages}

We now define the framework of information extraction introduced in Section~\ref{sec:formalFramework} purely in terms of formal languages. To this end, we will define a special finite alphabet $\Delta_{\Sigma, \Gamma}$ such that we can encode $(\Sigma, \Gamma)$-extractors as languages over $\Delta_{\Sigma, \Gamma}$.

Let $\Sigma$ and $\Gamma$ be fixed. For every set $X \subseteq \Gamma$ of attributes and $b \in \Sigma$, we call $X_b$ a \emph{$\Sigma$-signed $\Gamma$-marker}, where $\sign(X_b) := b$ is the \emph{sign} of $X_b$ and $X$ is the \emph{marker set} of $X_b$. We let $\Delta_{\Sigma, \Gamma}$ be the finite alphabet of all $\Sigma$-signed $\Gamma$-markers. By a slight abuse of notation, we also use the typical set notations directly on $\Sigma$-signed $\Gamma$-markers. In particular, we write $x \in X_b$ or $x \notin X_b$ to express $x \in X$ or $x \notin X$, respectively, or, for any $\odot \in \{\cup, \cap, \setminus\}$, we write $X_b \odot Y_b$ to denote $(X \odot Y)_b$. Strings over $\Delta_{\Sigma, \Gamma}$ are called \emph{$\Sigma$-signed $\Gamma$-marker strings}. For a $\Sigma$-signed $\Gamma$-marker string $W$ of length $n$, the \emph{sign} of $W$ is defined by $\sign(W) = \sign(W[1]) \sign(W[2]) \ldots \sign(W[n])$. 

\begin{example}
If $\Sigma = \{\ta, \tb, \tc, \td\}$ and $\Gamma = \{x, y, z\}$, then $\{x, z\}_{\ta} \: \emptyset_{\tb} \: \{x, y, z\}_{\tb} \: \{y\}_{\td}$ is a $\Sigma$-signed $\Gamma$-marker string with sign $\ta \tb \tb \td$.
\end{example}

For better readability, we also represented $\Sigma$-signed $\Gamma$-marker strings in a simplified way: we represent $\Sigma$-signed $\Gamma$-markers $\{x_1, \ldots, x_k\}_b$ by $[x_1, \ldots, x_k]_b$, and $\Sigma$-signed $\Gamma$-markers $\emptyset_b$ by the symbol $b$. For example, we write $\emptyset_{\ta} \: \emptyset_{\tb} \: \{x, z\}_{\ta} \: \{y\}_{\td} \: \emptyset_{\tc}$ as $\ta \: \tb \: [x, z]_{\ta} \: [y]_{\td} \: \tc$.

A set $L \subseteq (\Delta_{\Sigma, \Gamma})^*$ of $\Sigma$-signed $\Gamma$-marker strings is a \emph{$\Sigma$-signed $\Gamma$-marker} language. For a $\Sigma$-signed $\Gamma$-marker language $L$ and some $w \in \Sigma^*$, the set $\slice{w}(L) = \{W \in L \mid \sign(W) = w\}$ is called the \emph{$w$-signed slice of $L$}.

Next, we define a mapping $\encoding_{\Sigma, \Gamma}$ that encodes any pair $(w, t)$, where $w \in \Sigma^*$ and $t$ is a $\Gamma$-tuple for $w$, as the $\Sigma$-signed $\Gamma$-marker string of length $|w|$, where, for every $i \in \{1, 2, \ldots, |w|\}$, $\encoding_{\Sigma, \Gamma}(w, t)[i] = X_{w[i]}$ with $X = \{x \in \Gamma \mid i \in t(x)\}$. Note that this means that $\sign(\encoding_{\Sigma, \Gamma}(w, t)) = w$. We also lift $\encoding_{\Sigma, \Gamma}$ to a mapping that maps any $\Gamma$-table $T$ for some $w \in \Sigma^*$ to the finite set $\{\encoding_{\Sigma, \Gamma}(w, t) \mid t \in T\}$ of $\Sigma$-signed $\Gamma$-marker strings, and to a mapping that maps any $(\Sigma, \Gamma)$-extractor $E$ to the $\Sigma$-signed $\Gamma$-marker language $\{\encoding_{\Sigma, \Gamma}(w, t) \mid w \in \Sigma^*, t \in E(w)\}$. Note that for every $w \in \Sigma^*$, the $w$-signed slice of $\encoding_{\Sigma, \Gamma}(E)$ corresponds to the $\Gamma$-table $E(w)$, i.e., $\slice{w}(\encoding_{\Sigma, \Gamma}(E)) = \encoding_{\Sigma, \Gamma}(E(w))$. Moreover, we could also write $\encoding_{\Sigma, \Gamma}(E) = \bigcup_{w \in \Sigma^*} \encoding(E(w))$.

\begin{example}\label{ExtractorExample}
Let $\Sigma = \{\ta, \tb, \tc, \td\}$ and $\Gamma = \{x, y\}$. Let $E$ be the $\Sigma$-signed $\Gamma$-extractor that maps every $w \in \Sigma^*$ to the set of $\Gamma$-tuples with an $x$-entry $\{i\}$, where $w[i, i+1] = \ta \tb$ and an $y$-entry $\{j\}$, where $w[j, j+1] = \tc \tc$. For $w = \ta \tb \tc \ta \tb \tc \tc \tc$, we have that 
\begin{align*}
\encoding_{\Sigma, \Gamma}(E(w)) = \{\encoding_{\Sigma, \Gamma}(w, t) \mid t \in E(w)\} = \{\:& [x]_{\ta} \: \tb \: \tc \: \ta \: \tb \: [y]_{\tc} \: \tc \: \tc\\
&\ta \: \tb \: \tc \: [x]_{\ta} \: \tb \: [y]_{\tc} \: \tc \: \tc\\
&[x]_{\ta} \: \tb \: \tc \: \ta \: \tb \: \tc \: [y]_{\tc} \: \tc\\
&\ta \: \tb \: \tc \: [x]_{\ta} \: \tb \: \tc \: [y]_{\tc} \: \tc\}\,.
\end{align*}
Note that we use the simplified way of writing $\Sigma$-signed $\Gamma$-marker strings defined above.

The $\Sigma$-signed $\Gamma$-marker language $\encoding_{\Sigma, \Gamma}(E)$ is the following regular language 
\begin{equation*}
\{u \: [x]_{\ta} \: \tb \: v \: [x]_{\tc} \: \tc \: w \mid u, v, w \in \Sigma^*\} \cup \{u \: [x]_{\tc} \: \tc \: v \: [x]_{\ta} \: \tb \: w \mid u, v, w \in \Sigma^*\}\,.
\end{equation*}
\end{example}

The next proposition follows directly from the definitions.

\begin{proposition}
The mapping $\encoding_{\Sigma, \Gamma}$ is a bijection from 
\begin{itemize}
\item the set of pairs $(w, t)$ with $w \in \Sigma^*$ and $t$ a $\Gamma$-tuple for $w$ to the set $(\Delta_{\Sigma, \Gamma})^*$ of $\Sigma$-signed $\Gamma$-marker strings.
\item the set of $\Gamma$-tables for some $w \in \Sigma^*$ to the set of finite sets of $\Sigma$-signed $\Gamma$-marker strings with the same sign.
\item the set of $(\Sigma, \Gamma)$-extractors to the set of $\Sigma$-signed $\Gamma$-marker languages.
\end{itemize}
\end{proposition}

Technically, we could also consider extractors over arbitrary alphabets $\Sigma$ and $\Gamma$ (instead of fixed ones) and then define $\encoding(E) = \encoding_{\Sigma, \Gamma}(E)$ for every $(\Sigma, \Gamma)$-extractor $E$. However, this mapping $\encoding$ is not a bijection, because any $\Sigma$-signed $\Gamma$-marker language $L$ is also a $\Sigma'$-signed $\Gamma'$-marker language for every $\Sigma' \supseteq \Sigma$ and $\Gamma' \supseteq \Gamma$, and therefore $\encoding$ is not injective.

Let us next introduce some more convenient notations that shall be helpful in the remainder of this paper. For any signed marker string $W \in (\Delta_{\Sigma, \Gamma})^*$, we use $\llbracket W \rrbracket_{\Sigma, \Gamma}$ in order to denote the $\Gamma$-tuple described by $W$, i.e., $\encoding_{\Sigma, \Gamma}(\sign(W), \llbracket W \rrbracket_{\Sigma, \Gamma}) = W$. Likewise, for any $\Sigma$-signed $\Gamma$-marker language $L$, we use $\llbracket L \rrbracket_{\Sigma, \Gamma}$ to denote the $(\Sigma, \Gamma)$-extractor represented by $L$, i.e., $\encoding_{\Sigma, \Gamma}(E) = L$. In particular, this means that $\llbracket \slice{w}(L) \rrbracket_{\Sigma, \Gamma} = E(w)$ for every $w \in \Sigma^*$.

We can now conveniently define classes of extractors via the corresponding classes of signed marker language.

\begin{definition}\label{extractorClassDefinition}
Let $\mathcal{L}$ be a language class, let $\Sigma$ be a terminal alphabet and let $\Gamma$ be an attribute alphabet. Then $\mathcal{E}^{\Sigma, \Gamma}_{\mathcal{L}} = \{\llbracket L \rrbracket_{\Sigma, \Gamma} \mid L \in \mathcal{L} \wedge L \subseteq (\Delta_{\Sigma, \Gamma})^*\}$ is the set of $(\Sigma, \Gamma)$-extractors represented by $\Sigma$-signed $\Gamma$-marker languages from $\mathcal{L}$. We also define $\mathcal{E}^{\Sigma, \Gamma_{\downarrow}}_{\mathcal{L}} = \bigcup_{\Gamma' \subseteq \Gamma} \mathcal{E}^{\Sigma, \Gamma'}_{\mathcal{L}}$.
\end{definition}

For example, $\mathcal{E}^{\Sigma, \Gamma}_{\mathcal{L}}$ with $\mathcal{L} \in \{\reg, \cfl, \mathsf{CSL}, \mathsf{RE}\}$ are all well-defined classes of extractors, or we could consider the class of extractors represented by marker languages from the complexity classes NL, P, NP, etc. Extractors represented by undecidable marker languages are also covered by our definition.

\subsection{Operations on Signed Marker Languages}\label{sec:languageOperations}

We now define several operations on signed marker languages. More precisely, for a fixed terminal alphabet $\Sigma$, we define these operations for the class of $\Sigma$-signed $\Gamma$-marker languages for some attribute alphabet $\Gamma$. 

Let us first note that since signed marker languages are sets, their union, intersection and set difference are well-defined. Moreover, for a signed $\Gamma$-marker language $L$, we define $\overline{L}$ as its complement with respect to the universe $(\Delta_{\Sigma, \Gamma})^*$, i.e., $\overline{L} = (\Delta_{\Sigma, \Gamma})^* \setminus L$.

Next, we define the join variants (from Section~\ref{sec:joinvariants}) and the concatenation and the Kleene star (from Section~\ref{sec:concatKleene}). 

For a $\Gamma_1$-marker string $W_1$ and a $\Gamma_2$-marker string $W_2$, their concatenation $W_1 \cdot W_2$ is a $(\Gamma_1 \cup \Gamma_2)$-marker string with $\sign(W_1 \cdot W_2) = \sign(W_1) \cdot \sign(W_2)$. In order to define the different join variants for marker-strings, we first do this on the level of markers. Let $X_b \in \Delta_{\Gamma_1}$ and $Y_b \in \Delta_{\Gamma_2}$ (note that $\sign(X_b) = \sign(Y_b) = b$). Then, for every $\odot \in \{\cup, \cap, \setminus\}$, $X_b \bowtie_{\odot} Y_b$ is the $(\Gamma_1 \cup \Gamma_2)$-marker with sign $b$ and marker set $Z$, where,
\begin{itemize}
\item for every $x \in \Gamma_1 \setminus \Gamma_2$, $x \in Z \Leftrightarrow x \in X$,
\item for every $x \in \Gamma_2 \setminus \Gamma_1$, $x \in Z \Leftrightarrow x \in Y$,
\item for every $x \in \Gamma_1 \cap \Gamma_2$, $x \in Z \Leftrightarrow x \in X \odot Y$.
\end{itemize}

Next, let $W_1$ be a $\Gamma_1$-marker string, let $W_2$ be a $\Gamma_2$-marker string, and let $\sign(W_1) = \sign(W_2) = w \in \Sigma^*$. For every $\odot \in \{\cup, \cap, \setminus\}$, $W_1 \bowtie_{\odot} W_2$ is the $(\Gamma_1 \cup \Gamma_2)$-marker string with sign $w$, where, for every $i \in \{1, 2, \ldots, |w|\}$, $(W_1 \bowtie_{\odot} W_2)[i] = W_1[i] \bowtie_{\odot} W_2[i]$. 
Moreover, we define $W_1 \bowtie W_2$ as $W_1 \bowtie W_2 = \bot$ if there is some $i \in \{1, 2, \ldots, |w|\}$ such that $W_1[i] \cap \Gamma_1 \cap \Gamma_2 \neq W_2[i] \cap \Gamma_1 \cap \Gamma_2$(observe that this is equivalent to $\llbracket W_1 \rrbracket_{\Sigma, \Gamma_1}(x) \neq \llbracket W_2 \rrbracket_{\Sigma, \Gamma_2}(x)$ for some $x \in \Gamma_1 \cap \Gamma_2$), and $W_1 \bowtie W_2 = W_1 \bowtie_{\cup} W_2$ otherwise. Observe that $W_1 \circ W_2$ with $\circ \in \{\bowtie, \bowtie_{\cup}, \bowtie_{\cap}, \bowtie_{\setminus}\}$ is only defined if $\sign(W_1) = \sign(W_2)$, whereas $W_1 \cdot W_2$ is also defined in the case that $\sign(W_1) \neq \sign(W_2)$.

We can now lift these operations to marker languages. For any $\Gamma_1$-marker language $L_1$, any $\Gamma_2$-marker language $L_2$ and every $\circ \in \{\bowtie, \bowtie_{\cup}, \bowtie_{\cap}, \bowtie_{\setminus}\}$, let $L_1 \circ L_2 = \bigcup_{w \in \Sigma^*} \{W_1 \circ W_2 \mid W_1 \in \slice{w}(L_1), W_2 \in \slice{w}(L_2)\}$, and let $L_1 \cdot L_2 = \{W_1 \cdot W_2 \mid W_1 \in L_1, W_2 \in L_2\}$. Note that $L_1 \circ L_2$ and $L_1 \cdot L_2$ are signed $(\Gamma_1 \cup \Gamma_2)$-marker languages. Moreover, let the Kleene star be defined as the usual language operation, i.e., $L_1^* = \bigcup_{k \geq 0} L_1^k$.

\begin{example}
Let $\Gamma_1 = \{x, y\}$ and $\Gamma_2 = \{y, z\}$. Let $U_1 = \{x\}_\ta \{x\}_\tb \ta \{y\}_\tc \tb \ta \tc$ and $U_2 = \{y\}_\ta \tb \ta \tc \tb \ta \tc$ be $\Gamma_1$-marker strings, and let $V_1 = \ta \{x\}_\tb \{x, z\}_\ta \tc \tb \ta \{z\}_\tc$ and $V_2 = \{z\}_\ta \tb \{z\}_\ta \{x\}_\tc \tb \ta \tc$ be $\Gamma_2$-marker strings. Note that $\sign(U_1) = \sign(U_2) = \sign(V_1) = \sign(V_2) = \ta \tb \ta \tc \tb \ta \tc$. We observe that $U_1 \bowtie_{\cup} V_1 = \{x\}_\ta \{x\}_\tb \{x, z\}_\ta \{y\}_\tc \tb \ta \{z\}_\tc$ and $U_1 \bowtie_{\cap} V_1 = \ta \{x\}_\tb \{z\}_\ta \{y\}_\tc \tb \ta \{z\}_\tc$. Moreover, 
\begin{align*}
\{U_1, U_2\} \bowtie_{\setminus} \{V_1, V_2\} = \{&U_1 \bowtie_{\setminus} V_1, U_1 \bowtie_{\setminus} V_2, U_2 \bowtie_{\setminus} V_1, U_2 \bowtie_{\setminus} V_2\}\\
= \{&\{x\}_\ta \tb \{z\}_\ta \{y\}_\tc \tb \ta \{z\}_\tc, \{x, z\}_\ta \{x \}_\tb \{z \}_\ta \{y \}_\tc \tb \ta \tc, \\
&\{y \}_\ta \tb \{z \}_\ta \tc \tb \ta \{z \}_\tc, \{y, z \}_\ta \tb \{z \}_\ta \tc \tb \ta \tc\}\,.
\end{align*}
Note that this example refers to the example discussed at the end of Section~\ref{sec:joinvariants} about the corresponding operations on the level of marker tuples and tables.
\end{example}

Next, we will define the operations on marker strings and marker languages that correspond to the unary operations of Section~\ref{sec:otherOperators}. Let $W$ be some $\Gamma$-marker string with sign $w \in \Sigma^*$, let $\Gamma' \subseteq \Gamma$, let $y, y' \in \Gamma$ and let $z \notin \Gamma$. Then $\pi_{\Gamma'}(W)$ is the $\Gamma'$-marker string with sign $w$, where, for every $i \in \{1, 2, \ldots, |w|\}$, $(\pi_{\Gamma'}(W))[i] = W[i] \cap \Gamma'$.

For every $\odot \in \{\cup, \cap, \setminus\}$, $\colmerge_{y, y', \odot}(W)$ is the $(\Gamma \setminus \{y'\})$-marker string with sign $w$, where, for every $i \in \{1, 2, \ldots, |w|\}$ and $x \in \Gamma \setminus \{y, y'\}$, $x \in (\colmerge_{y, y', \odot}(W))[i] \Leftrightarrow x \in W[i]$ and
\begin{itemize}
\item $y \in (\colmerge_{y, y', \cup}(W))[i] \Leftrightarrow y \in W[i] \vee y' \in W(i)$,
\item $y \in (\colmerge_{y, y', \cap}(W))[i] \Leftrightarrow y \in W[i] \wedge y' \in W(i)$,
\item $y \in (\colmerge_{y, y', \setminus}(W))[i] \Leftrightarrow y \in W[i] \wedge y' \notin W(i)$.
\end{itemize}

Finally, $\colrename_{y \to z}(W)$ is the $((\Gamma \setminus \{y\}) \cup \{z\})$-marker string with sign $w$, where, for every $i \in \{1, 2, \ldots, n\}$ and $x \in \Gamma \setminus \{y\}$, $x \in (\colrename_{y \to z}(W))[i] \Leftrightarrow x \in W[i]$, and $z \in (\colrename_{y \to z}(W))[i] \Leftrightarrow y \in W[i]$.

Next, we lift these operations to marker languages. For any marker language $L$ and every $f \in \{\pi_{\Gamma'}, \colmerge_{y, y', \cup}, \colmerge_{y, y', \cap}, \colmerge_{y, y', \setminus}, \colrename_{y \to z}\}$, let $f(L) = \{f(W) \mid W \in L\}$. 

\subsection{Homomorphism Between Marker Languages and Extractors}\label{sec:isomorphism}

In Section~\ref{sec:operations}, we have defined several operations on extractors and in Section~\ref{sec:languageOperations}, we have defined the same operations on marker languages. It follows from the definitions of these operations that the interpretation $\llbracket \cdot \rrbracket_{\Sigma, \Gamma}$ of $\Sigma$-signed $\Gamma$-marker languages as $(\Sigma, \Gamma)$-extractors is a homomorphism with respect to these operations. This means that in order to compose extractors by our operations or to use them in order to define more complex extractors from basic extractors, we can instead work on the level of signed marker languages.

Let us first consider the Boolean set operations.

\begin{theorem}\label{homoSetOpsTheorem}
Let $\Sigma$ be a terminal alphabet and let $\Gamma$ be an attribute alphabet. For every $L_1, L_2 \in (\Delta_{\Sigma, \Gamma})^*$ and $\odot \in \{\cup, \cap, \setminus\}$, we have that $\llbracket L_1 \odot L_2 \rrbracket_{\Sigma, \Gamma} = \llbracket L_1 \rrbracket_{\Sigma, \Gamma} \odot \llbracket L_2 \rrbracket_{\Sigma, \Gamma}$; moreover, $\llbracket \overline{L} \rrbracket_{\Sigma, \Gamma} = \neg \llbracket L \rrbracket_{\Sigma, \Gamma}$ and $\llbracket L^* \rrbracket_{\Sigma, \Gamma} = (\llbracket L \rrbracket_{\Sigma, \Gamma})^*$.
\end{theorem}

Next, we observe the analogous statement for the concatenation and the join variants. Note that these operations are defined for extractors and signed marker languages over possibly different attribute alphabets.

\begin{theorem}\label{homoJoinTheorem}
Let $\Sigma$ be a terminal alphabet and let $\Gamma_1$ and $\Gamma_2$ be attribute alphabets. For every $L_1 \in (\Delta_{\Sigma, \Gamma_1})^*$, $L_2 \in (\Delta_{\Sigma, \Gamma_2})^*$ and $\odot \in \{\cdot, \bowtie, \bowtie_{\cup}, \bowtie_{\cap}, \bowtie_{\setminus}\}$, we have that $\llbracket L_1 \odot L_2 \rrbracket_{\Sigma, \Gamma_1 \cup \Gamma_2} = \llbracket L_1 \rrbracket_{\Sigma, \Gamma_1} \odot \llbracket L_2 \rrbracket_{\Sigma, \Gamma_2}$.
\end{theorem}

Finally, we consider the special unary operations.

\begin{theorem}\label{homoUnaryTheorem}
Let $\Sigma$ be a terminal alphabet and let $\Gamma$ be an attribute alphabet. For every $\Gamma' \subseteq \Gamma$, $y, y' \in \Gamma$, and $z \notin \Gamma$, we have that $\pi_{\Gamma'}(L)$ is a $\Sigma$-signed $\Gamma'$-marker language with $\llbracket \pi_{\Gamma'}(L) \rrbracket_{\Sigma, \Gamma'} = \pi_{\Gamma'}(\llbracket L \rrbracket_{\Sigma, \Gamma})$, $f(L)$ is a $\Sigma$-signed $(\Gamma \setminus \{y'\})$-marker language with $\llbracket f(L) \rrbracket_{\Sigma, \Gamma \setminus \{y'\}} = f(\llbracket L \rrbracket)_{\Sigma, \Gamma}$ for every $f \in \{\colmerge_{y, y', \cup}, \colmerge_{y, y', \cap}, \colmerge_{y, y', \setminus}\}$, and $\colrename_{y \to y'}(L)$ is a $\Sigma$-signed $((\Gamma \setminus \{y\}) \cup \{y'\})$-marker language with $\llbracket \colrename_{y \to y'}(L) \rrbracket_{\Sigma, (\Gamma \setminus \{y\}) \cup \{y'\}} = \colrename_{y \to y'}(\llbracket L \rrbracket_{\Sigma, \Gamma'})$.
\end{theorem}

These results follow more or less directly from the definitions and the proofs are straightforward. Therefore, we shall give them in~\ref{sec:IsomorphismAppendix}.

\section{Regular and Context-Free Extractors}\label{sec:regCFExtractors}

In this section (as well as in Section~\ref{sec:compProblems}), we consider the classes $\mathcal{E}^{\Sigma, \Gamma}_{\reg}$ and $\mathcal{E}^{\Sigma, \Gamma}_{\cfl}$ of \emph{regular} and \emph{context-free $(\Sigma, \Gamma)$-extractors} (see Definition~\ref{extractorClassDefinition}), i.e., extractors represented by regular and context-free signed marker languages, which can therefore be represented by NFAs and CFGs (and other equivalent description mechanisms). Those classes inherit several nice properties from the classes of regular and context-free languages. 

\subsection{Closure Properties of Regular Extractors}\label{sec:closurePropRegExtractors}

Regular extractors are closed under the extractor operations mentioned in Section~\ref{sec:operations}, which can be shown by using the interpretation of extractor operations as operations on signed marker languages (see Section~\ref{sec:languageOperations}) and then showing that regular signed marker languages are closed under these operations. 

Recall that for any alphabet $A$, we denote the class of regular and context-free languages over $A$ by $\reg_A$ and $\cfl_A$, respectively. In particular, $\reg_{\Delta_{\Sigma, \Gamma}}$ and $\cfl_{\Delta_{\Sigma, \Gamma}}$ are the classes of regular and context-free $\Sigma$-signed $\Gamma$-marker languages, respectively. According to Definition~\ref{extractorClassDefinition}, $\mathcal{E}^{\Sigma, \Gamma}_{\reg} = \{\llbracket L \rrbracket_{\Sigma, \Gamma} \mid L \in \reg_{\Delta_{\Sigma, \Gamma}}\}$ and $\mathcal{E}^{\Sigma, \Gamma}_{\cfl} = \{\llbracket L \rrbracket_{\Sigma, \Gamma} \mid L \in \cfl_{\Delta_{\Sigma, \Gamma}}\}$.

The closure properties of regular marker languages under Boolean operations, concatenation and Kleene star directly follow from the known closure properties of regular languages:

\begin{proposition}\label{simpleLanguageClosurePropertiesProposition}
If $L_1, L_2 \in \reg_{\Delta_{\Sigma, \Gamma}}$, then the languages $L_1 \cup L_2$, $L_1 \cap L_2$, $L_1 \setminus L_2$, $\overline{L_1}$, $L_1 \cdot L_2$ and $(L_1)^*$ are in $\reg_{\Delta_{\Sigma, \Gamma}}$.
\end{proposition}

Proving this for the different join variants is also not difficult, if we argue on the level of NFAs.

\begin{proposition}\label{joinVariantsLanguageClosurePropertiesProposition}
Let $L_1 \in \reg_{\Delta_{\Sigma, \Gamma_1}}$, $L_2 \in \reg_{\Delta_{\Sigma, \Gamma_2}}$ and $\circ \in \{\bowtie, \bowtie_{\cup}, \bowtie_{\cap}, \bowtie_{\setminus}\}$. Then $L_1 \circ L_2 \in \reg_{\Delta_{\Sigma, \Gamma_1 \cup \Gamma_2}}$. 
\end{proposition}

\begin{proof}
Let $M_1$ and $M_2$ be NFAs that accept $L_1$ and $L_2$, respectively. We will define an NFA that accepts $L_1 \circ L_2$, which proves the claim of the proposition.

Recall that for every $\odot \in \{\cup, \cap, \setminus\}$ we have defined $X_b \bowtie_{\odot} Y_b$ for every $\Sigma$-signed $\Gamma_1$-marker $X_b$ and every $\Sigma$-signed $\Gamma_2$-marker $Y_b$ (see Section~\ref{sec:languageOperations}). Moreover, $W \in L_1 \circ L_2$ if and only if there is a $\Gamma_1$-marker string $W_1$ and a $\Gamma_2$-marker string $W_2$ with $\sign(W_1) = \sign(W_2) = \sign(W)$ such that $W[i] = W_1[i] \bowtie_{\odot} W_2[i]$ for every $i \in \{1, 2, \ldots, |W|\}$. 

Let $M_{\bowtie_{\odot}}$ be an NFA that has $Q_1 \times Q_2$ as state set, where $Q_1$ and $Q_2$ are the state sets of $M_1$ and $M_2$, respectively. The transition function of $M_{\bowtie_{\odot}}$ is defined as follows. For every states $p_1, p_2 \in Q_1$ and $q_1, q_2 \in Q_2$, $M_1$ has an $X_b$-labelled transition from state $p_1$ to $p_2$ and $M_2$ has a $Y_b$-labelled transition from state $q_1$ to $q_2$ if and only if $M_{\bowtie_{\odot}}$ has a $X_b \bowtie_{\odot} Y_b$-labelled transition from state $(p_1, q_1)$ to $(p_2, q_2)$. We let $(q_{0, 1}, q_{0, 2})$ be the initial state of $M_{\bowtie_{\odot}}$ (where $q_{0, 1}$ and $q_{0, 2}$ are the initial states of $M_1$ and $M_2$, respectively) and we let a state $(p, q)$ be accepting in $M_{\bowtie_{\odot}}$ if $p$ is accepting in $M_1$ and $q$ is accepting in $M_2$.

We observe that a $\Sigma$-signed $(\Gamma_1 \cup \Gamma_2)$-marker string $W$ is accepted by $M_{\bowtie_{\odot}}$ if and only if $M_1$ accepts a $\Gamma_1$-marker string $W_1$, $M_2$ accepts a $\Gamma_2$-marker string $W_2$ with $\sign(W_1) = \sign(W_2) = \sign(W)$ and such that $W[i] = W_1[i] \bowtie_{\odot} W_2[i]$ for every $i \in \{1, 2, \ldots, |W|\}$. Consequently, $M_{\bowtie_{\odot}}$ accepts $L_1 \circ L_2$. 

It remains to prove that $L_1 \bowtie L_2 \in \reg_{\Delta_{\Sigma, \Gamma_1 \cup \Gamma_2}}$. A NFA $M_{\bowtie}$ that accepts $L_1 \bowtie L_2$ can be constructed quite similar to the automaton $M_{\bowtie_{\cup}}$. The transition function of $M_{\bowtie}$ is defined as for $M_{\bowtie_{\cup}}$, with the only difference that a $X_b \bowtie_{\odot} Y_b$-labelled transition from state $(p_1, q_1)$ to $(p_2, q_2)$ only exists if the corresponding $X_b$- and $Y_b$-labelled transitions from $M_1$ and $M_2$ additionally satisfy that $X_b \cap \Gamma_1 \cap \Gamma_2 = Y_b \cap \Gamma_1 \cap \Gamma_2$.
\end{proof}

Next, we consider the unary operators from Section~\ref{sec:otherOperators}.

\begin{proposition}\label{unaryLanguageClosurePropertiesProposition}
Let $L \in \reg_{\Delta_{\Sigma, \Gamma}}$, let $\Gamma' \subseteq \Gamma$, let $y, y' \in \Gamma$, and let $z \notin \Gamma$. Then we have that $\pi_{\Gamma'}(L) \in \reg_{\Delta_{\Sigma, \Gamma'}}$, $f(L) \in \reg_{\Delta_{\Sigma, \Gamma \setminus \{y'\}}}$ for every $f \in \{\colmerge_{y, y', \cup}$, $\colmerge_{y, y', \cap}$, $\colmerge_{y, y', \setminus}\}$, and $\colrename_{y \to z}(L) \in \reg_{\Delta_{\Sigma, (\Gamma \setminus \{y\}) \cup \{z\}}}$.
\end{proposition}

\begin{proof}
By definition, the operation $\pi_{\Gamma'}(\cdot)$ simply removes all occurrences of markers from $\Gamma \setminus \Gamma'$ from the $\Sigma$-signed $\Gamma$-markers of a $\Sigma$-signed $\Gamma$-marker language. Thus, for any NFA $M$ that accepts a $\Sigma$-signed $\Gamma$-marker language, we can easily construct an NFA $M'$ with $L(M') = \pi_{\Gamma'}(L(M))$, which shows that $\pi_{\Gamma'}(L(M)) \in \reg_{\Delta_{\Sigma, \Gamma'}}$.

By definition, the operation $f(\cdot)$ for every $f \in \{\colmerge_{y, y', \cup}, \colmerge_{y, y', \cap}, \colmerge_{y, y', \setminus}\}$ changes a $\Gamma$-marker language only insofar as $y'$ is removed from every $\Gamma$-marker, and $y$ stays in any $\Gamma$-marker $X_b$ only if $y \in X$ or $y' \in X$ (case $f = \colmerge_{y, y', \cup}$), or only if $y \in X$ and $y' \in X$ (case $f = \colmerge_{y, y', \cap}$), or only if $y \in X$ and $y' \notin X$ (case $f = \colmerge_{y, y', \setminus}$). Again, it can be easily seen that any NFA $M$ that accepts a $\Sigma$-signed $\Gamma$-marker language can be transformed into an NFA $M'$ with $L(M') = f(L(M))$ for every $f \in \{\colmerge_{y, y', \cup}, \colmerge_{y, y', \cap}, \colmerge_{y, y', \setminus}\}$.

Finally, the operation $\colrename_{y \to z}(\cdot)$ simply renames every $y$ in some $\Sigma$-signed $\Gamma$-marker of some $\Sigma$-signed $\Gamma$-marker language into $z$. Again, we can directly conclude that any NFA $M$ that accepts a $\Sigma$-signed $\Gamma$-marker language can be transformed into an NFA $M'$ with $L(M') = \colrename_{y \to z}(L(M))$. 
\end{proof}

Since the representation of extractors by signed marker languages is a homomorphism with respect to the considered operations (see Section~\ref{sec:isomorphism}), the closure properties above with respect to regular signed marker languages directly imply the following closure properties for regular extractors.

\begin{theorem}\label{closurePropForRegExtractorsProposition}
Let $E_1, E_2 \in \mathcal{E}^{\Sigma, \Gamma}_{\reg}$. Then $E_1 \odot E_2 \in \mathcal{E}^{\Sigma, \Gamma_1 \cup \Gamma_2}_{\reg}$, for every $\odot \in \{\cup, \cap, \setminus\}$. Moreover, $\neg{E_1}, (E_1)^* \in \mathcal{E}^{\Sigma, \Gamma}_{\reg}$.
\end{theorem}

\begin{theorem}\label{closurePropForRegExtractorsProposition}
Let $E_1 \in \mathcal{E}^{\Sigma, \Gamma_1}_{\reg}$ and $E_2 \in \mathcal{E}^{\Sigma, \Gamma_2}_{\reg}$. Then $E_1 \odot E_2 \in \mathcal{E}^{\Sigma, \Gamma_1 \cup \Gamma_2}_{\reg}$, for every $\odot \in \{\cdot, \bowtie, \bowtie_{\cup}, \bowtie_{\cap}, \bowtie_{\setminus}\}$. 
\end{theorem}

\begin{theorem}\label{closurePropForRegExtractorsPropositionUnaryOps}
Let $E \in \mathcal{E}^{\Sigma, \Gamma}_{\reg}$, let $\Gamma' \subseteq \Gamma$, let $y, y' \in \Gamma$, and let $z \notin \Gamma$. Then $\pi_{\Gamma'}(E) \in \mathcal{E}^{\Sigma, \Gamma'}_{\reg}$, $f(E) \in \mathcal{E}^{\Sigma, \Gamma \setminus \{y'\}}_{\reg}$ for every $f \in \{\colmerge_{y, y', \cup}, \colmerge_{y, y', \cap}, \colmerge_{y, y', \setminus}\}$, and $\colrename_{y \to z}(E) \in \mathcal{E}^{\Sigma, (\Gamma \setminus \{y\}) \cup \{z\}}_{\reg}$.
\end{theorem}

\subsection{Representations of Regular Extractors}\label{sec:regularRepresentations}

Regular extractors can be represented by any kind of regular language description mechanisms (since we only have to represent a regular signed marker language). However, there are two other natural ways of representing regular extractors -- as algebraic expressions over atomic extractors, and as MSO-formulas over $\Sigma$-strings.

Let us recall that for a terminal alphabet $\Sigma$ and an attribute alphabet $\Gamma$, the term $(\Sigma, \Gamma_{\downarrow})$-extractor refers to any $(\Sigma, \Gamma')$-extractor with $\Gamma' \subseteq \Gamma$. Moreover, the set $\mathcal{E}^{\Sigma, \Gamma_{\downarrow}}_{\mathcal{L}}$ contains all $(\Sigma, \Gamma_{\downarrow})$-extractors that can be represented by a regular $\Sigma$-signed $\Gamma'$-marker language with $\Gamma' \subseteq \Gamma$, i.e., $\mathcal{E}^{\Sigma, \Gamma_{\downarrow}}_{\mathcal{L}} = \bigcup_{\Gamma' \subseteq \Gamma} \mathcal{E}^{\Sigma, \Gamma'}_{\mathcal{L}}$. 

For an arbitrary set of $(\Sigma, \Gamma_{\downarrow})$-extractors $\mathcal{E}$ and a subset $\Lambda$ of the extractor operations defined in Section~\ref{sec:operations}, a \emph{$\Lambda$-expression over atoms from $\mathcal{E}$} is a valid algebraic expression that uses operations from $\Lambda$ and atoms from $\mathcal{E}$ (here ``valid'' means that every operation in the expression is well-defined for its arguments, and its result is a $(\Sigma, \Gamma_{\downarrow})$-extractor). By $\mathcal{E}^{\Lambda}$ we denote the set of all $(\Sigma, \Gamma_{\downarrow})$-extractors that can be described by a $\Lambda$-expression over atoms from $\mathcal{E}$. Let $\Lambda_{\textsf{full}}$ be the set of most of the extractor-operators discussed in Section~\ref{sec:operations}, i.e., the set operations $\cup, \cap, \setminus, \neg$, the join variants $\bowtie, \bowtie_{\cup}, \bowtie_{\cap}, \bowtie_{\setminus}$, concatenation, Kleene star, merge and projection.\footnote{We exclude the renaming operation, since it adds a new element to the set $\Gamma$ of attributes instead of restricting it to a subset, which does not fit to the statement of the Proposition~\ref{characterisationProp}.}

Due to the closure properties discussed in Section~\ref{sec:closurePropRegExtractors}, we know that:

\begin{proposition}\label{characterisationProp}
$(\mathcal{E}^{\Sigma, \Gamma_{\downarrow}}_{\reg})^{\Lambda_{\textsf{full}}} = \mathcal{E}^{\Sigma, \Gamma_{\downarrow}}_{\reg}$.
\end{proposition}

For every signed marker $X_b \in \Delta_{\Sigma, \Gamma}$ with $b \in \Sigma$, we consider $\llbracket \{X_b\} \rrbracket_{\Sigma, \Gamma}$ as an \emph{atomic $(\Sigma, \Gamma)$-extractors}. Recall that $(\llbracket \{X_b\} \rrbracket_{\Sigma, \Gamma})(b) = t$ with $t(x) = \{1\}$ for every $x \in X$ and $t(y) = \emptyset$ for every $y \in \Gamma \setminus X$. Let $\mathbb{A}_{\Sigma, \Gamma} = \{\llbracket \{X_b\} \rrbracket_{\Sigma, \Gamma} \mid X_b \in \Delta_{\Sigma, \Gamma}\}$ be the set of all atomic $(\Sigma, \Gamma)$-extractors.

By definition, every regular $(\Sigma, \Gamma)$-extractor can be represented by a regular $\Sigma$-signed $\Gamma$-marker language, which can be represented by a regular expression over the alphabet $\Delta_{\Sigma, \Gamma}$. Moreover, such a regular expression over $\Delta_{\Sigma, \Gamma}$ translates into a valid algebraic expression that uses extractor-operations $\cup$, $\cdot$ and $*$ and that has atomic $(\Sigma, \Gamma)$-extractors as atoms. This yields the following lemma.

\begin{lemma}\label{algebraicCharacterisationRegularExtractorsLemma}
$(\mathbb{A}_{\Sigma, \Gamma})^{\{\cup, \cdot, *\}} = \mathcal{E}^{\Sigma, \Gamma}_{\reg}$.
\end{lemma}

\begin{proof}
The inclusion $(\mathbb{A}_{\Sigma, \Gamma})^{\{\cup, \cdot, *\}} \subseteq \mathcal{E}^{\Sigma, \Gamma}_{\reg}$ is obviously true, since every $\llbracket \{X_b\} \rrbracket_{\Sigma, \Gamma} \in \mathbb{A}_{\Sigma, \Gamma}$ is a regular $(\Sigma, \Gamma)$-extractor (described by the singleton $\Sigma$-signed $\Gamma$-marker language $\{X_b\}$) and the regular $(\Sigma, \Gamma)$-extractors are closed under union, concatenation and Kleene star.

Let us show the other direction. To this end, let $E \in \mathcal{E}^{\Sigma, \Gamma}_{\reg}$ be some $(\Sigma, \Gamma)$-extractor, and let $L$ be the $\Sigma$-signed $\Gamma$-marker language of $E$, i.e., $E = \llbracket L \rrbracket_{\Sigma, \Gamma}$. Since $E$ is regular, $L$ is a regular language over the alphabet $\Delta_{\Sigma, \Gamma}$. Consequently, $L$ can be described by an $\{\cup, \cdot, *\}$-expression over atoms of the form $\{X_b\}$ with $X_b \in \Delta_{\Sigma, \Gamma}$ (i.e., a regular expression). Now by replacing in this $\{\cup, \cdot, *\}$-expression each atom $\{X_b\}$ by the atomic $(\Sigma, \Gamma)$-extractor $\llbracket \{X_b\} \rrbracket_{\Sigma, \Gamma}$, we obtain a $\{\cup, \cdot, *\}$-expression over atoms from $\mathbb{A}_{\Sigma, \Gamma}$. Inductively applying Theorems~\ref{homoSetOpsTheorem}~and~\ref{homoJoinTheorem} bottom-up shows that this $\{\cup, \cdot, *\}$-expression over atoms from $\mathbb{A}_{\Sigma, \Gamma}$ evaluates to $\llbracket L \rrbracket_{\Sigma, \Gamma} = E$. 
\end{proof}

Let us now come to the representation by formulas of monadic second order logic ($\mso$ for short). To this end, we interpret strings over some alphabet $A$ as relational structures in the usual way, i.e., as relational structures $w = (\{1, 2, \ldots, n\}, <, (P_{a})_{a \in A})$, where $<$ is the linear order on $\{1, 2, \ldots, n\}$ and the $P_a$ are unary relations that describe a partition of $\{1, 2, \ldots, n\}$, i.e., $P_a \cap P_{a'} = \emptyset$ for every $a, a' \in A$ with $a \neq a'$, and $\bigcup_{a \in A} P_a = \{1, 2, \ldots, n\}$.

For any alphabet $A$, formulas of \emph{monadic second order logic for $A$-strings} ($\mso_{A}$) are just $\mso$-formulas over the signature $(<, (P_{a})_{a \in A})$. For an $\mso_{A}$-sentence $\phi$ and a string $w \in A^*$, we write $w \models \phi$ to denote that $\phi$ holds in the string $w$, and $L(\phi) = \{w \mid w \models \phi\}$ is the language over $A$ described by $\phi$. For an $\mso_{A}$-formula $\phi(X_1, X_2, \ldots, X_k)$, a string $w$ and sets $S_1, S_2, \ldots, S_k \subseteq \{1, 2, \ldots, |w|\}$, we write $w \models \phi(S_1, \ldots, S_k)$ to denote that the formula $\phi$ holds in the string $w$ if the variable $X_i$ is set to $S_i$ for every $i \in \{1, 2, \ldots, k\}$. Hence, we can interpret each $\mso_{A}$-formula $\phi(X_1, \ldots, X_k)$ as defining a \emph{result set} $\phi(w) = \{(S_1, S_2, \ldots, S_k) \mid w \models \phi(S_1, \ldots, S_k)\}$. In the following, we consider monadic second order logic for $\Sigma$-strings and $\Delta_{\Sigma, \Gamma}$-strings (i.e., marker strings). For more details about MSO logic and its relevance for formal languages we refer to the textbook~\cite[Chapter $7$]{Libkin2004}.

It is well-known that a language $L \subseteq A^*$ is regular if and only if there is an $\mso_{A}$-sentence $\phi$ with $L = L(\phi)$ (see, e.\,g.,~\cite[Chapter $7$]{Libkin2004}). Consequently, we can also describe regular $(\Sigma, \Gamma)$-extractors by $\mso_{\Delta_{\Sigma, \Gamma}}$-sentences (since they represent the regular $\Sigma$-signed $\Gamma$-marker languages). However, we can also show that regular $(\Sigma, \Gamma)$-extractors can directly be represented by $\mso_{\Sigma}$-formulas (with a free set variable per attribute of $\Gamma$).

Let $\Gamma = \{\gamma_1, \gamma_2, \ldots, \gamma_m\}$ be an attribute alphabet. The \emph{$\Gamma$-$\mso_{\Sigma}$ formulas} are the $\mso_{\Sigma}$ formulas of the form $\phi(X_{\gamma_1}, X_{\gamma_2}, \ldots, X_{\gamma_m})$, i.e., $\mso_{\Sigma}$ formulas with exactly one free set-variable for each attribute. Since every $(S_1, S_2, \ldots, S_k) \in \phi(w)$ for a $\Gamma$-$\mso_{\Sigma}$ formula $\phi(X_{\gamma_1}, X_{\gamma_2}, \ldots, X_{\gamma_m})$ can be interpreted as a $\Gamma$-tuple, the result set $\phi(w)$ can be interpreted as a $\Gamma$-table. Consequently, any $\Gamma$-$\mso_{\Sigma}$ formula $\phi(X_{\gamma_1}, X_{\gamma_2}, \ldots, X_{\gamma_m})$ represents a $(\Sigma, \Gamma)$-extractor $\llbracket \phi \rrbracket_{\Sigma, \Gamma}$. Finally, $\mathcal{E}^{\Sigma, \Gamma}_{\mso}$ denotes the set of $(\Sigma, \Gamma)$-extractors described by $\Gamma$-$\mso_{\Sigma}$ formulas. 

In order to see that $\Gamma$-$\mso_{\Sigma}$ formulas describe exactly the set of regular $(\Sigma, \Gamma)$-extractors, we first have to show that any such formula $\phi$ can be transformed into a sentence that describes the signed marker language of the extractor $\llbracket \phi \rrbracket_{\Sigma, \Gamma}$, and vice versa. Then, the well-known Büchi–Elgot–Trakhtenbrot Theorem~\cite[Chapter $7$]{Libkin2004} implies that the signed marker language is necessarily regular.

We first prove two lemmas.

\begin{lemma}\label{fromFormulasToSentencesLemma}
For every $\Gamma$-$\mso_{\Sigma}$ formula $\phi(X_{\gamma_1}, \ldots, X_{\gamma_m})$, we can construct an $\mso_{\Delta_{\Sigma, \Gamma}}$ sentence $\phi'$ such that $L(\phi')$ satisfies $\llbracket L(\phi') \rrbracket_{\Sigma, \Gamma} = \llbracket \phi \rrbracket_{\Sigma, \Gamma}$.
\end{lemma}

\begin{proof}
We define 
\begin{equation*}
\phi' = \exists X_{\gamma_1}, \ldots, \exists X_{\gamma_m}: \psi(X_{\gamma_1}, \ldots, X_{\gamma_m}) \wedge \phi''(X_{\gamma_1}, \ldots, X_{\gamma_m})\,,
\end{equation*}
where $\psi(X_{\gamma_1}, \ldots, X_{\gamma_m})$ is an $\mso_{\Delta_{\Sigma, \Gamma}}$ formula such that $W \models \psi(S_1, S_2, \ldots, S_m)$ if and only if $(S_1, S_2, \ldots, S_m) = \llbracket W \rrbracket_{\Sigma, \Gamma}$, and $\phi''$ is obtained from $\phi$ by replacing each atom $P_{b}(i)$ by the expression $\bigvee_{Y \subseteq \Gamma} P_{Y_b}(i)$. Note that the formula $\psi$ can be easily constructed, and that $\phi'$ is indeed an $\mso_{\Delta_{\Sigma, \Gamma}}$ sentence. 

Now recall that $\encoding_{\Sigma, \Gamma}$ is the encoding of pairs of a word over $\Sigma$ and a $\Gamma$-tuple for $w$ into a $\Sigma$-signed $\Gamma$-marker string. We will next prove the following claim. 

\medskip

\noindent \emph{Claim}: For every $w \in \Sigma^*$ and $\Gamma$-tuple $(S_1, S_2, \ldots, S_m)$ for $w$, we have that 
\begin{equation*}
w \models \phi(S_1, S_2, \ldots, S_m) \iff \encoding_{\Sigma, \Gamma}(w, (S_1, \ldots, S_m)) \models \phi'\,.
\end{equation*}

\smallskip

\noindent \emph{Proof of claim}: Let $w$ be a word over $\Sigma$ and let $(S_1, S_2, \ldots, S_m)$ be a $\Gamma$-tuple for $w$. We first assume that $w \models \phi(S_1, S_2, \ldots, S_m)$. We observe that $\encoding_{\Sigma, \Gamma}(w, (S_1, \ldots, S_m)) \models \psi(S_{1}, \ldots, S_{m})$ since $(S_1, S_2, \ldots, S_m) = \llbracket \encoding_{\Sigma, \Gamma}(w, (S_1, \ldots, S_m)) \rrbracket_{\Sigma, \Gamma}$. Moreover, since $w = \sign(\encoding_{\Sigma, \Gamma}(w, (S_1, \ldots, S_m)))$, we have that $w \models P_{b}(i)$ if and only if $\encoding_{\Sigma, \Gamma}(w, (S_1, \ldots, S_m)) \models \bigvee_{Y \subseteq \Gamma} P_{Y_b}(i)$. This implies that $\encoding_{\Sigma, \Gamma}(w, (S_1, \ldots, S_m)) \models \phi''(S_1, \ldots, S_m)$. We can therefore conclude that $\encoding_{\Sigma, \Gamma}(w, (S_1, \ldots, S_m)) \models \phi'$.\par
Let us next assume that $\encoding_{\Sigma, \Gamma}(w, (S_1, \ldots, S_m)) \models \phi'$. We observe that this implies that there is a $\Gamma$-tuple $(S'_1, \ldots, S'_m)$ for $w$ with $\encoding_{\Sigma, \Gamma}(w, (S_1, \ldots, S_m)) \models \psi(S'_1, \ldots, S'_m)$ and $\encoding_{\Sigma, \Gamma}(w, (S_1, \ldots, S_m)) \models \phi''(S'_1, \ldots, S'_m)$. But $\encoding_{\Sigma, \Gamma}(w, (S_1, \ldots, S_m)) \models \psi(S'_1, \ldots, S'_m)$ implies that $(S_1, \ldots, S_m) = (S'_1, \ldots, S'_m)$, so we conclude that $\encoding_{\Sigma, \Gamma}(w, (S_1, \ldots, S_m)) \models \phi''(S_1, \ldots, S_m)$. As before, $w \models P_{b}(i)$ if and only if $\encoding_{\Sigma, \Gamma}(w, (S_1, \ldots, S_m)) \models \bigvee_{Y \subseteq \Gamma} P_{Y_b}(i)$, which means that $w \models \phi(S_1, \ldots, S_m)$. \par
This concludes the proof of the claim.

\medskip

We can now directly conclude that $\llbracket L(\phi') \rrbracket_{\Sigma, \Gamma} = \llbracket \phi \rrbracket_{\Sigma, \Gamma}$. To this end, let $w \in \Sigma^*$. If $(S_1, \ldots, S_m) \in \llbracket \phi \rrbracket_{\Sigma, \Gamma}(w)$, then $w \models \phi(S_1, S_2, \ldots, S_m)$, which, by the claim above, means that $\encoding_{\Sigma, \Gamma}(w, (S_1, \ldots, S_m)) \models \phi'$. This implies that $\encoding_{\Sigma, \Gamma}(w, (S_1, \ldots, S_m)) \in L(\phi')$ and therefore $(S_1, \ldots, S_m) \in \llbracket L(\phi') \rrbracket_{\Sigma, \Gamma}(w)$. On the other hand, if $(S_1, \ldots, S_m) \in \llbracket L(\phi') \rrbracket_{\Sigma, \Gamma}(w)$, then $\encoding_{\Sigma, \Gamma}(w, (S_1, \ldots, S_m)) \in L(\phi')$, which, by the claim above, means that $w \models \phi(S_1, S_2, \ldots, S_m)$. Thus, $(S_1, \ldots, S_m) \in \llbracket \phi \rrbracket_{\Sigma, \Gamma}(w)$. 
\end{proof}

\begin{lemma}\label{fromSentencesToFormulasLemma}
For every $\mso_{\Delta_{\Sigma, \Gamma}}$ sentence $\phi$, a $\Gamma$-$\mso_{\Sigma}$ formula $\phi'(X_{\gamma_1}, \ldots, X_{\gamma_m})$ can be constructed that satisfies $\llbracket \phi' \rrbracket_{\Sigma, \Gamma} = \llbracket L(\phi) \rrbracket_{\Sigma, \Gamma}$.
\end{lemma}

\begin{proof}
Let $\phi$ be an $\mso_{\Delta_{\Sigma, \Gamma}}$ sentence. Let $\phi'(X_{\gamma_1}, \ldots, X_{\gamma_m})$ be the $\Gamma$-$\mso_{\Sigma}$ formula obtained from $\phi$ by replacing every atom $P_{Y_b}(i)$ in $\phi$ by the formula 
\begin{equation*}
\psi_{i, Y, b}(X_{\gamma_1}, \ldots, X_{\gamma_m}) = P_b(i) \wedge (\bigwedge_{\gamma_j \in Y} i \in X_{\gamma_j}) \wedge (\bigwedge_{\gamma_j \in \Gamma \setminus Y} i \notin X_{\gamma_j})\,. 
\end{equation*}

\medskip

\noindent \emph{Claim}: For every $w \in \Sigma^*$ and every $\Gamma$-tuple $(S_1, S_2, \ldots, S_m)$ for $w$, we have that 
\begin{equation*}
\encoding_{\Sigma, \Gamma}(w, (S_1, \ldots, S_m)) \models \phi \iff w \models \phi'(S_1, S_2, \ldots, S_m)\,.
\end{equation*}

\smallskip

\noindent \emph{Proof of claim}: Let $w \in \Sigma^*$ and let $(S_1, S_2, \ldots, S_m)$ be a $\Gamma$-tuple for $w$. We observe that $\encoding_{\Sigma, \Gamma}(w, (S_1, \ldots, S_m)) \models P_{Y_b}(i)$ if and only if $w \models \psi_{i, Y, b}(S_1, \ldots, S_m)$, which directly yields the statement of the claim. 

\medskip

We can now directly conclude that $\llbracket L(\phi) \rrbracket_{\Sigma, \Gamma} = \llbracket \phi' \rrbracket_{\Sigma, \Gamma}$. To this end, let $w \in \Sigma^*$. If $(S_1, \ldots, S_m) \in \llbracket L(\phi) \rrbracket_{\Sigma, \Gamma}(w)$, then $\encoding_{\Sigma, \Gamma}(w, (S_1, \ldots, S_m)) \in L(\phi)$, which means that $\encoding_{\Sigma, \Gamma}(w, (S_1, \ldots, S_m)) \models \phi$. By the claim from above, this implies that we have $w \models \phi'(S_1, S_2, \ldots, S_m)$. Thus, $(S_1, S_2, \ldots, S_m) \in \llbracket \phi' \rrbracket_{\Sigma, \Gamma}(w)$. If, on the other hand, we have that $(S_1, \ldots, S_m) \in \llbracket \phi' \rrbracket_{\Sigma, \Gamma}(w)$, then $w \models \phi'(S_1, S_2, \ldots, S_m)$, which, by the claim above, means that we have $\encoding_{\Sigma, \Gamma}(w, (S_1, \ldots, S_m)) \models \phi$. Hence, $\encoding_{\Sigma, \Gamma}(w, (S_1, \ldots, S_m)) \in L(\phi)$ and therefore $(S_1, \ldots, S_m) \in \llbracket L(\phi) \rrbracket_{\Sigma, \Gamma}(w)$. 
\end{proof}

Now we can show that $\Gamma$-$\mso_{\Sigma}$ formulas describe exactly the set of regular $\Gamma$-extractors.

\begin{lemma}\label{MSOandRegEqualityLemma}
$\mathcal{E}^{\Sigma, \Gamma}_{\reg} = \mathcal{E}^{\Sigma, \Gamma}_{\mso}$. 
\end{lemma}

\begin{proof}
Let $E \in \mathcal{E}^{\Sigma, \Gamma}_{\reg}$, which means that there is a regular $\Sigma$-signed $\Gamma$-marker language $L$ with $E = \llbracket L \rrbracket_{\Sigma, \Gamma}$. By the Büchi–Elgot–Trakhtenbrot theorem~\cite[Chapter $7$]{Libkin2004}, there is an $\mso_{\Delta_{\Sigma, \Gamma}}$ sentence $\phi$ with $L(\phi) = L$. By Lemma~\ref{fromSentencesToFormulasLemma}, we can transform $\phi$ into a $\Gamma$-$\mso_{\Sigma}$ formula $\phi'$ with $\llbracket L(\phi) \rrbracket_{\Sigma, \Gamma} = \llbracket \phi' \rrbracket_{\Sigma, \Gamma}$. We conclude: $E = \llbracket L \rrbracket_{\Sigma, \Gamma} = \llbracket L(\phi) \rrbracket_{\Sigma, \Gamma} = \llbracket \phi' \rrbracket_{\Sigma, \Gamma}$ for a $\Gamma$-$\mso_{\Sigma}$ formula $\phi'$, which means that $E \in \mathcal{E}^{\Sigma, \Gamma}_{\mso}$.

Let $E \in \mathcal{E}^{\Sigma, \Gamma}_{\mso}$, so there is a $\Gamma$-$\mso_{\Sigma}$ formula $\phi$ with $E = \llbracket \phi \rrbracket_{\Sigma, \Gamma}$. By Lemma~\ref{fromFormulasToSentencesLemma}, we can transform $\phi$ into an $\mso_{\Delta_{\Sigma, \Gamma}}$ sentence $\phi'$ such that $\llbracket L(\phi') \rrbracket_{\Sigma, \Gamma} = \llbracket \phi \rrbracket_{\Sigma, \Gamma}$. Again by the Büchi–Elgot–Trakhtenbrot theorem, we know that there is a regular $\Sigma$-signed $\Gamma$-marker language with $L = L(\phi')$. We conclude: $E = \llbracket \phi \rrbracket_{\Sigma, \Gamma} = \llbracket L(\phi') \rrbracket_{\Sigma, \Gamma} = \llbracket L \rrbracket_{\Sigma, \Gamma}$ for a regular $\Sigma$-signed $\Gamma$-marker language $L$, which means that $E \in \mathcal{E}^{\Sigma, \Gamma}_{\reg}$. 
\end{proof}

Hence, we have shown that regular extractors are also characterised by MSO-formulas and algebraic expressions using union, concatenation and Kleene star.

\begin{theorem}
$\mathcal{E}^{\Sigma, \Gamma}_{\reg} = \mathcal{E}^{\Sigma, \Gamma}_{\mso} = (\mathbb{A}_{\Sigma, \Gamma})^{\{\cup, \cdot, *\}}$.
\end{theorem}

\subsection{Closure Properties of Context-Free Extractors}

From the known closure properties of context-free languages, we can conclude the following closure properties of context-free extractors.

\begin{proposition}\label{CFLClosureProposition}
For every $E_1, E_2 \in \mathcal{E}^{\Sigma, \Gamma}_{\cfl}$ we have that $E_1 \cup E_2, (E_1)^* \in \mathcal{E}^{\Sigma, \Gamma}_{\cfl}$. For every $E_1 \in \mathcal{E}^{\Sigma, \Gamma_1}_{\cfl}$ and $E_2 \in \mathcal{E}^{\Sigma, \Gamma_2}_{\cfl}$, we have that $E_1 \cdot E_2 \in \mathcal{E}^{\Sigma, \Gamma_1 \cup \Gamma_2}_{\cfl}$.
\end{proposition}

\begin{proof}
That $L_1 \cup L_2, (L_1)^* \in \cfl_{\Delta_{\Sigma, \Gamma}}$ for every $L_1, L_2 \in \cfl_{\Delta_{\Sigma, \Gamma}}$ and that $L_1 \cdot L_2 \in \cfl_{\Delta_{\Sigma, \Gamma_1 \cup \Gamma_2}}$ for every $L_1 \in \cfl_{\Delta_{\Sigma, \Gamma_1}}$ and $L_2 \in \cfl_{\Delta_{\Sigma, \Gamma_2}}$ follows directly from the known closure properties of context-free language. Hence, the statement of the proposition follows then with Propositions~\ref{languageOperationsBooleanProposition}~and~\ref{binaryLanguageOperationsLanguagesProposition}.
\end{proof}

The non-closure of context-free languages under intersection and complement implies the following non-closure properties for context-free extractors.

\begin{proposition}\label{CFLNonClosureProposition}
Let $E_1, E_2 \in \mathcal{E}^{\Sigma, \Gamma}_{\cfl}$. Then $E_1 \cap E_2$, $E_1 \setminus E_2$ and $\neg E_1$ are not necessarily context-free. Let $E_1 \in \mathcal{E}^{\Sigma, \Gamma_1}_{\cfl}$ and $E_2 \in \mathcal{E}^{\Sigma, \Gamma_2}_{\cfl}$. Then $E_1 \circ E_2$ for $\circ \in \{\bowtie, \bowtie_{\cup}, \bowtie_{\cap}, \bowtie_{\setminus}\}$ are not necessarily context-free. 
\end{proposition}

\begin{proof}
Let $L_1$ and $L_2$ be languages over $\{\emptyset_b \mid b \in \Sigma\}$, i.e., the marker set of every marker is the empty set. Note that this means that $L_1$ and $L_2$ are $\Sigma$-signed $\Gamma$-marker languages for any arbitrary attribute alphabet $\Gamma$. By Proposition~\ref{languageOperationsBooleanProposition}, $\llbracket L_1 \rrbracket_{\Sigma, \Gamma} \cap \llbracket L_2 \rrbracket_{\Sigma, \Gamma}= \llbracket L_1 \cap L_2 \rrbracket_{\Sigma, \Gamma}$, $\llbracket L_1 \rrbracket_{\Sigma, \Gamma}\setminus \llbracket L_2 \rrbracket_{\Sigma, \Gamma}= \llbracket L_1 \setminus L_2 \rrbracket_{\Sigma, \Gamma}$, and $\neg \llbracket L_1 \rrbracket_{\Sigma, \Gamma}= \llbracket \overline{L_1} \rrbracket_{\Sigma, \Gamma}$. Furthermore, by Proposition~\ref{binaryLanguageOperationsLanguagesProposition}, $\llbracket L_1 \rrbracket_{\Sigma, \Gamma}\circ \llbracket L_2 \rrbracket_{\Sigma, \Gamma}= \llbracket L_1 \circ L_2 \rrbracket_{\Sigma, \Gamma}$ for every $\circ \in \{\bowtie, \bowtie_{\cup}, \bowtie_{\cap}, \bowtie_{\setminus}\}$.

It is a well-known fact that for context-free languages $L_1$ and $L_2$ the languages $L_1 \cap L_2$, $L_1 \setminus L_2$ and $\overline{L_1}$ are not necessarily context-free; thus, for such $L_1$ and $L_2$, the extractors $\llbracket L_1 \cap L_2 \rrbracket_{\Sigma, \Gamma}$, $\llbracket L_1 \setminus L_2 \rrbracket_{\Sigma, \Gamma}$ and $\llbracket \overline{L_1} \rrbracket_{\Sigma, \Gamma}$ are not necessarily context-free.

By definition, for every $\circ \in \{\bowtie, \bowtie_{\cup}, \bowtie_{\cap}, \bowtie_{\setminus}\}$ we have that $\emptyset_b \circ \emptyset_b = \emptyset_b$ for every $b \in \Sigma$, and $L_1 \circ L_2 = \bigcup_{w \in \Sigma^*} \{W_1 \circ W_2 \mid W_1 \in \slice{w}(L_1), W_2 \in \slice{w}(L_2)\}$. Since for every $w = a_1 \ldots a_n \in \Sigma^*$ either $\slice{w}(L_1) = \{\emptyset_{a_1} \ldots \emptyset_{a_n}\}$ or $\slice{w}(L_1) = \emptyset$ (and analogously for $L_2$), we can conclude that $L_1 \circ L_2 = L_1 \cap L_2$. Consequently, for $L_1$ and $L_2$ with $L_1 \cap L_2$ not context-free, we can conclude that $\llbracket L_1 \rrbracket_{\Sigma, \Gamma} \circ \llbracket L_2 \rrbracket_{\Sigma, \Gamma} = \llbracket L_1 \circ L_2 \rrbracket_{\Sigma, \Gamma} = \llbracket L_1 \cap L_2 \rrbracket_{\Sigma, \Gamma}$ is not a context-free extractor, for $\circ \in \{\bowtie, \bowtie_{\cup}, \bowtie_{\cap}, \bowtie_{\setminus}\}$. 
\end{proof}

Finally, we observe that the context-free extractors are closed under the unary operators from Section~\ref{sec:otherOperators}.

\begin{proposition}\label{CFLClosurePropositionUnary}
Let $E \in \mathcal{E}^{\Sigma, \Gamma}_{\cfl}$, let $\Gamma' \subseteq \Gamma$, let $y, y' \in \Gamma$, and let $z \notin \Gamma$. Then $\pi_{\Gamma'}(E) \in \mathcal{E}^{\Sigma, \Gamma'}_{\cfl}$, $f(E) \in \mathcal{E}^{\Sigma, \Gamma \setminus \{y'\}}_{\cfl}$ for every $f \in \{\colmerge_{y, y', \cup}, \colmerge_{y, y', \cap}, \colmerge_{y, y', \setminus}\}$, and $\colrename_{y \to z}(E) \in \mathcal{E}^{\Sigma, (\Gamma \setminus \{y\}) \cup \{z\}}_{\cfl}$.
\end{proposition}

\begin{proof}
Let $L \in \cfl_{\Delta_{\Sigma, \Gamma}}$. Recall that the operation $\pi_{\Gamma'}(\cdot)$ simply removes all occurrences of markers from $\Gamma \setminus \Gamma'$ from $L$; the operation $f(\cdot)$ for every $f \in \{\colmerge_{y, y', \cup}, \colmerge_{y, y', \cap}, \colmerge_{y, y', \setminus}\}$ changes $L$ only insofar as $y'$ is removed from every $\Sigma$-signed $\Gamma$-marker, and $y$ stays in any $\Sigma$-signed $\Gamma$-marker $X_b$ only if $y \in X$ or $y' \in X$ (case $f = \colmerge_{y, y', \cup}$), or only if $y \in X$ and $y' \in X$ (case $f = \colmerge_{y, y', \cap}$), or only if $y \in X$ and $y' \notin X$ (case $f = \colmerge_{y, y', \setminus}$); and the operation $\colrename_{y \to z}(\cdot)$ simply renames every $y$ in some $\Sigma$-signed $\Gamma$-marker of $L$ into $z$. Consequently, by accordingly manipulating all $\Sigma$-signed $\Gamma$-markers of the context-free grammar of $L$, we obtain a context-free grammar for the $\Sigma$-signed $\Gamma'$-marker language $\pi_{\Gamma'}(L)$, a context-free grammar for the $\Sigma$-signed $(\Gamma \setminus \{y'\})$-marker language $f(L)$ for every $f \in \{\colmerge_{y, y', \cup}, \colmerge_{y, y', \cap}, \colmerge_{y, y', \setminus}\}$, and a context-free grammar for the $\Sigma$-signed $((\Gamma \setminus \{y\}) \cup \{z\})$-marker language $\colrename_{y \to z}(L)$. 
\end{proof}

\section{Computational Problems}\label{sec:compProblems}

In this section, we need some more details and conventions about finite automata. In general, we write NFA as $M = (Q, \Sigma, \delta, q_0, F)$, where $Q$ is the set of states, $\Sigma$ is the input alphabet, $\delta : Q \times \Sigma \to \powerset{Q}$ is the transition function (or $\delta : Q \times \Sigma \to Q$ if $M$ is deterministic), $q_0$ is the start state and $F$ is the set of final states. As a convention, for every $i \in \mathbb{N}$, we write $M_i$ in order to denote an NFA of the form $M_i = (Q_i, \Sigma, \delta_i, q_{0, i}, F_i)$.

Context-free grammars are tuples $G = (V, \Sigma, P, S)$, where $V$ is the set of non-terminals, $\Sigma$ is the set of terminals, $P \subseteq V \times (V \cup \Sigma)^*$ is the set of rules and $S$ is the start non-terminal. We denote rules $(A, v) \in S$ also by $A \to v$.

We will now investigate certain natural computational problems for extractors, and we will mainly concentrate on regular and context-free extractors, which, if not stated otherwise, are represented as NFAs and CFGs, respectively. For convenience, we simply write $|E|$ to denote the size of the NFA or CFG that represents the extractor $E$. We first discuss several problems that trivially reduce to well-known formal language problems on the marker languages. 

The \emph{tuple membership problem} (for a class $\mathcal{E}$ of extractors) is to decide for a given extractor $E \in \mathcal{E}$, a string $w$ and a tuple $t$ for $w$ whether $t \in E(w)$.

Obviously, checking $t \in E(w)$ boils down to checking $\encoding_{\Sigma, \Gamma}(w, t) \in \encoding_{\Sigma, \Gamma}(E)$, which means that the tuple membership problem for regular or context-free extractors inherits the upper bounds of the membership problem for regular or context-free languages (and similarly for any other language class). Moreover, we can interpret any $L \subseteq \Sigma^*$ as a $\Sigma$-signed $\emptyset$-marker language and then check $w \in L$ by checking $t^{\emptyset}_{\Gamma} \in \llbracket L \rrbracket_{\Sigma, \Gamma}(w)$. Hence, the conditional lower bounds for the membership problem for regular and context-free languages (see~\cite{BackursIndyk2016,AbboudEtAl2018}) carry over to the tuple membership problem for regular and context-free extractors.

In principle, we could also compute the full set $E(w)$ by testing $\encoding_{\Sigma, \Gamma}(w, t) \in \encoding_{\Sigma, \Gamma}(E)$ for every $\Gamma$-tuple $t$ for $w$. This is of course rather inefficient, since there are $2^{|w| |\Gamma|}$ different $\Gamma$-tuples for $w$ that we have to consider. However, for regular and context-free extractors, we can do much better, which has been thoroughly investigated in terms of an enumeration problem (i.e., we wish to enumerate all elements from $E(w)$ without repetition and with a guaranteed upper bound on the delay between two elements). See, e.\,g.,~\cite{AmarilliEtAl2021,GawrychowskiEtAl2024,MunozRiveros2025} for regular extractors and~\cite{AmarilliEtAl2022} for context-free extractors (in addition, there are several papers investigating the enumeration problem for regular and context-free document spanners (e.\,g.~\cite{Peterfreund2019PhD,Peterfreund2023,FlorenzanoEtAl2020}), which are a subset of our regular and context-free extractors). For self-containment, we shall discuss in Section~\ref{sec:enum} below these known results and how they apply to the framework presented here.

Let us move on to more complex decision problems. For $(\Sigma, \Gamma)$-extractors $E_1, E_2$, we write $E_1 \subseteq E_2$ if and only if $E_1(w) \subseteq E_2(w)$ for every $w \in \Sigma^*$ (note that $E_1 = E_2 \iff E_1 \subseteq E_2 \wedge E_2 \subseteq E_1$). The \emph{containment} and \emph{equivalence problem} is to decide for given $(\Sigma, \Gamma)$-extractors $E_1, E_2$ whether $E_1 \subseteq E_2$ or $E_1 = E_2$, respectively, and the \emph{emptiness problem} is to decide for a given $(\Sigma, \Gamma)$-extractor $E$ whether there exists some $w \in \Sigma^*$ with $E(w) \neq \emptyset$ (note that this must not be confused with checking whether an extractor is the empty extractor, which does not map the empty string to $\emptyset$; see Section~\ref{sec:formalFramework}). These decision problems are obviously identical to the corresponding problems on the marker languages.

\begin{observation}
For $(\Sigma, \Gamma)$-marker languages $L_1$ and $L_2$, we have that $\llbracket L_1 \rrbracket_{\Sigma, \Gamma} \subseteq \llbracket L_2 \rrbracket_{\Sigma, \Gamma}$ if and only if $L_1 \subseteq L_2$, $\llbracket L_1 \rrbracket_{\Sigma, \Gamma} = \llbracket L_2 \rrbracket_{\Sigma, \Gamma}$ if and only if $L_1 = L_2$, and $\llbracket L_1 \rrbracket_{\Sigma, \Gamma}(w) \neq \emptyset$ for some $w$ if and only if $L_1 \neq \emptyset$.
\end{observation}

Hence, from the well-known complexities of the classical decision problems for regular and context-free languages (see, e.g.,~\cite{HopcroftEtAl,StockmeyerMeyer1973}), we can conclude that the containment problem for regular extractors is PSPACE-complete, the emptiness problem for regular extractors is in P, the equivalence problem of context-free extractors is undecidable, etc.

\subsection{Enumeration of the Table}\label{sec:enum}

Given a $(\Sigma, \Gamma)$-extractor $E$ and a string $w \in \Sigma^*$, it is a reasonable task to compute the full table $E(w)$. However, since this table can be rather large, it makes sense to enumerate it row by row, i.e., we aim at an algorithm that first performs some preprocessing on $E$ and $w$, and then it produces $E(w)$ row by row without repetitions and with a guaranteed upper bound on the delay between two rows. In the database theory perspective, the best-case scenario is that this can be done with a preprocessing of $O(|w| f(|E|))$ and a delay of $O(f(|E|) |t|)$, where $t$ is the next row to produce, for some computable function $f$. This perspective covers the practically reasonable assumption that $w$ represents data and is therefore very large, while $E$ is a query of comparatively small size (after all, $E$ is assumed to be formulated by and should be understandable by human users). In other words, the preprocessing should be linear in $|w|$, while the delay should be linear in the size of the output element (i.e. a row of the output table) that is produced. This complexity measure that neglects the query size $|E|$ is called \emph{data complexity} in the database theory literature, and an algorithm with the aforementioned guarantees is called a linear preprocessing and output-linear delay algorithm. 

For the extractor framework of regular document spanners (which is covered by our extractor model), linear preprocessing and output-linear delay algorithms have been developed over the last decade. Note that since in the document spanner framework the rows of the output table have constant size, these algorithms are called linear preprocessing and \emph{constant} delay algorithms. However, in our more general setting, a row of the output table contains sets of positions as entries and therefore can have size $|w| |\Gamma|$, which depends on the data size $|w|$; thus, linear preprocessing and output-linear delay algorithms are the best we can hope for. It follows more or less directly from known results that for regular $(\Sigma, \Gamma)$-extractors linear preprocessing and output-linear delay algorithms for enumerating the output table exist. Let us explain this in some more detail.

In~\cite{Bagan2006}, it is shown that the result set of an MSO-query over a labelled tree can be enumerated with linear preprocessing and output-linear delay (in data complexity).\footnote{Note that~\cite{LohreySchmid2026}, which extends the algorithm of~\cite{Bagan2006} to the compressed setting, also contains an explanation of the approach of~\cite{Bagan2006}.} By our observations of Section~\ref{sec:regularRepresentations} (especially Lemma~\ref{MSOandRegEqualityLemma}), we know that we can transform any $(\Sigma, \Gamma)$-extractor into an MSO-formula $\phi$ such that, for every $w \in \Sigma^*$, the result set of $\phi$ on $w$ coincides with $E(w)$. Since we can interpret any string $w \in \Sigma^*$ as a labelled tree, the algorithm of~\cite{Bagan2006} can be used in order to enumerate $E(w)$ with linear preprocessing and output-linear delay (in data complexity). Let us next discuss a more direct approach.

Observe that a row of our output table $E(w)$, i.e., a $\Gamma$-tuple $t$ for $w$, can as well be represented as a sequence of pairs $(i, X)$ where $i \in \{1, 2, \ldots, |w|\}$ and $X$ is a non-empty subset of $\Gamma$ as follows: We simply move through $i = 1, 2, \ldots, |w|$ and if $\{x \mid i \in t(x)\} \neq \emptyset$, then we append the pair $(i, \{x \mid i \in t(x)\})$ to our sequence. For example, the first row of the table in Section~\ref{sec:intuition} is represented as $(1, \{z\}) (3, \{z\}) (4, \{x\}) (7, \{x\}) (8, \{y\}) (9, \{x, y\})$, and the third row is represented as $(2, \{y\}) (5, \{y, z\})$. This representation can also be seen as a condensed way of writing the $\Sigma$-signed $\Gamma$-marker string $\encoding_{\Sigma, \Gamma}(t)$, i.e., we delete markers of the form $\emptyset_{b}$ and represent markers $X_b$ as $(i, X)$, where $i$ is the position of the marker in the signed marker string. Obviously, this sequence representation uniquely describes $\Gamma$-tuples and has the same size as the $\Gamma$-tuples. Hence, enumerating these sequences is as good as enumerating the actual rows of the table. Let us also observe that these sequences are elements of the free monoid generated by the symbols $\{1, 2, \ldots, |w|\} \times (\powerset{\Gamma} \setminus \{\emptyset\})$, and with an identity element $\eword$.

Now if we are given an NFA $M$ for a $\Sigma$-signed $\Gamma$-marker language $L$, then the task of enumerating the sequences corresponding to $\llbracket L \rrbracket_{\Sigma, \Gamma}(w)$ boils down to enumerating the labels of certain paths in a DAG whose edges are labelled by elements from $(\{1, 2, \ldots, |w|\} \times (\powerset{\Gamma} \setminus \{\emptyset\})) \cup \{\eword\}$. More precisely, we can combine $M$ and $w$ into a DAG with nodes of the form $(q, i)$, where $q$ is a state of $M$ and $i \in \{0, 1, \ldots, |w|\}$ (similar to the product automaton construction with $w$ interpreted as an automaton). Moreover, if $M$ has a transition from state $p$ to state $q$ labelled with $X_{w[i]}$ and $X \neq \emptyset$, then we add a $(i, X)$-labelled arc to the DAG from $(p, i-1)$ to $(q, i)$, and if it has a transition $\emptyset_{w[i]}$, then we add an arc from $(p, i-1)$ to $(q, i)$ labelled with the identity element $\eword$. We can also add a unique source node $s$ (i.e., $s$ has no incoming edges) and a unique sink node $t$ (i.e., $t$ has no outgoing edges), such that the labels from the $s$-to-$t$ paths correspond to the desired output sequences that represent the rows of $\llbracket L \rrbracket_{\Sigma, \Gamma}(w)$. Consequently, we just have to enumerate all labels from $s$-to-$t$ paths in this DAG, which is not trivial due to the identity element $\eword$ as arc label (so we cannot naively enumerate the actual paths, since they might be much longer than the actual labels). 

There exist algorithmic techniques to efficiently enumerate all the labels of the $s$-to-$t$ paths in this DAG. For example, \cite[Section~$3$]{LohreySchmid2026} contains a general meta-theorem for enumerating path labels of a DAG whose nodes are labelled by the objects and edges are labelled by the morphisms of a category. A version of this (described in~\cite[Section~$3.3.1$]{LohreySchmid2026}) can be used for the enumeration of the labels of the $s$-to-$t$ paths as described above.\footnote{As a small caveat, the technique of~\cite{LohreySchmid2026} requires $M$ to be unambiguous, i.e., no accepted input can be accepted by two different accepting computations. However, if $M$ is not unambiguous, then we can make it first unambiguous in time that only depends on $|M|$ and therefore vanishes in the data complexity perspective.} Moreover, the general algorithmic technique that has been used in~\cite{AmarilliEtAl2021} in the context of regular document spanners can also be used in order to enumerate the labels of the $s$-to-$t$ paths of the DAG described above (more precisely, the algorithm described in Sections~$3$,~$4$~and~$5$ of~\cite{AmarilliEtAl2021} can be applied).

\subsection{Table Problems}\label{sec:tableProbs}

Let us next define the so-called \emph{table problems}. The \emph{table containment}, \emph{table equivalence} and \emph{table disjointness problem} is to decide for given $(\Sigma, \Gamma)$-extractors $E_1, E_2$ and $w \in \Sigma^*$ whether $E_1(w) \subseteq E_2(w)$, $E_1(w) = E_2(w)$ or $E_1(w) \cap E_2(w) \neq \emptyset$, respectively, and the \emph{table emptiness problem} is to decide for a given $(\Sigma, \Gamma)$-extractor $E$ and $w \in \Sigma^*$ whether $E(w) = \emptyset$. 

Unlike the problems from above, the table problems are not already covered by known language problems on the marker languages. Instead, they can be seen as problems on the $w$-slice of a marker language $L$, i.e., the set $\slice{w}(L) = \{W \in L \mid \sign(W) = w\}$ (see Section~\ref{sec:extractorsFormalLanguages}). More precisely, $E_1(w) \subseteq E_2(w)$, $E_1(w) = E_2(w)$, $E_1(w) \cap E_2(w) \neq \emptyset$ and $E(w) = \emptyset$ if and only if $\slice{w}(L_{E_1}) \subseteq \slice{w}(L_{E_2})$, $\slice{w}(L_{E_1}) = \slice{w}(L_{E_2})$, $\slice{w}(L_{E_1}) \cap \slice{w}(L_{E_2}) \neq \emptyset$ and $\slice{w}(L_{E}) = \emptyset$, respectively. 

The slices of any marker language are always finite sets of strings, but this does not necessarily mean that the table problems are easy, since slices have in general exponential size. For the mere decidability of the table problems, the decidability of the tuple membership problem (which we can assume for most reasonable classes of extractors) is a sufficient condition.

\begin{theorem}
If the tuple membership problem is decidable for a class $\mathcal{E}$ of extractors, then the table problems for $\mathcal{E}$ are decidable.
\end{theorem}

We now investigate the table problems for regular and context-free extractors.

\subsubsection{Regular Extractors}

We first observe that the table disjointness and emptiness problems can be solved efficiently by exploiting the NFA representation.

\begin{theorem}\label{regTableDisjointmentTheorem}
For regular extractors, the table disjointness and table emptiness problem can be solved in time $O(|E_1| |E_2| |w|)$ and $O(|E| |w|)$, respectively.
\end{theorem}

\begin{proof}
Let us start with the table disjointness problem. Let $E_1, E_2$ be regular extractors represented by NFAs $M_1$ and $M_2$, and let $w$ be a string. We construct a DAG $G_{M_1, M_2, w}$ with nodes $(p_1, p_2, i)$ for every $p_1 \in Q_1, p_2 \in Q_2$ and $i \in \{0, 1, \ldots, w\}$, and there is an arc from $(p_1, p_2, i)$ to $(q_1, q_2, i+1)$ if $q_1 \in \delta_1(p_1, X_{w[i+1]})$ and $q_2 \in \delta_2(p_2, X_{w[i+1]})$ for some $X \subseteq \Gamma$. It is obvious that $G_{M_1, M_2, w}$ can be constructed in time $O(|M_1| |M_2| |w|)$ and has size $O(|M_1| |M_2| |w|)$. We observe that there is a path from $(q_{0, 1}, q_{0, 2}, 0)$ to some $(q_{f, 1}, q_{f, 2}, |w|)$ with $q_{f, 1} \in F_1$ and $q_{f, 2} \in F_2$ if and only if there is a marker string $W \in L(M_1) \cap L(M_2)$ with $\sign(W) = w$. This latter property is characteristic for $E_1(w) \cap E_2(w) \neq \emptyset$. We can check whether such a path exists in time $O(|G_{M_1, M_2, w}|) = O(|M_1| |M_2| |w|)$.\par

Now let us consider the table emptiness problem. Let $E$ be a regular extractor represented by an NFA $M$ and let $w$ be a string. We construct a DAG $G_{M, w}$ with nodes $(p, i)$ for every $p \in Q$ and $i \in \{0, 1, \ldots, w\}$, and there is an arc from $(p, i)$ to $(q, i+1)$ if $q \in \delta(p, X_{w[i+1]})$ for some $X \subseteq \Gamma$. It is obvious that $G_{M, w}$ can be constructed in time $O(|M| |w|)$ and has size $O(|M| |w|)$. Similar as before, we can observe that there is a path from $(q_0, 0)$ to some $(q_f, |w|)$ with $q_f \in F$ if and only if there is a marker string $W \in L(M)$ with $\sign(W) = w$, which is characteristic for $E(w) \neq \emptyset$. We can check whether such a path exists in time $O(|G_{M, w}|) = O(|M| |w|)$.
\end{proof}

We can complement the upper bounds of Theorem~\ref{regTableDisjointmentTheorem} with conditional lower bounds. To this end, first recall that, conditional to the strong exponential time hypothesis (SETH), the membership problem for NFAs cannot be solved in time $O((|M||w|)^{1-\epsilon})$ for any $\epsilon > 0$ (see~\cite{BackursIndyk2016}). For a given NFA $M$ and string $w$, $w \in L(M)$ if and only if $\llbracket L(\widehat{M}) \rrbracket_{\Sigma, \Gamma}(w) \neq \emptyset$, where $\widehat{M}$ is obtained from $M$ by replacing every $b$-transition by a $\emptyset_b$-transition. Hence, $\llbracket E \rrbracket_{\Sigma, \Gamma}(w) \neq \emptyset$ cannot be checked in time $O((|E| |w|)^{1 - \epsilon})$, unless SETH fails. Likewise, $w \in L(M)$ if and only if $\llbracket L(\widehat{M}) \rrbracket_{\Sigma, \Gamma}(w) \cap \llbracket L(\widehat{M}') \rrbracket_{\Sigma, \Gamma}(w) \neq \emptyset$, where $L(\widehat{M}') = (\{\emptyset_b \mid b \in \Sigma\})^*$. Since $|\widehat{M}'| = |\Sigma| = O(1)$ (recall that we made the assumption that $\Sigma$ is of constant size), this means that if $E_1(w) \cap E_2(w) \neq \emptyset$ can be checked in time $O((|E_1| |E_2| |w|)^{1-\epsilon})$, then we can check $w \in L(M)$ in time $O((|M| |M'| |w|)^{1-\epsilon}) = O((|M| |w|)^{1-\epsilon})$, which contradicts SETH.

In contrast to table disjointness and emptiness, the table containment and equivalence problems are intractable. However, the complexity changes if extractors are represented by DFAs instead of NFAs. Recall that a $(\Sigma, \Gamma)$-extractor $E$ has finite support if the set $\{w \in \Sigma^* \mid E(w) \neq \emptyset\}$ is finite.

\begin{theorem}\label{coNPHardnessRegularContainment}
The table containment and table equivalence problem for regular extractors is coNP-complete, even if $|\Sigma| = 1$, $|\Gamma| = 2$, both $E_1$ and $E_2$ have finite support, and $E_1$ is given by a DFA. The table containment problem for regular extractors can be solved in time $O(|E_1| |E_2| |w|)$, provided that $E_2$ is given by a DFA, and the table equivalence problem for regular extractors can be solved in time $O(|E_1| |E_2| |w|)$, provided that both $E_1$ and $E_2$ are given by DFAs.
\end{theorem}

\begin{proof}
The coNP-membership of the table containment problem is obvious: Let $M_1$ and $M_2$ be NFA that represent regular extractors $E_1$ and $E_2$, respectively. Then we can guess a marker string $W$ with $\sign(W) = w$ (note that for this we only have to guess $|w|$ markers, each of which can be guessed by $|\Gamma|$ guesses) and then check whether $W \in L(M_1)$ and $W \notin L(M_2)$, which is characteristic for $E_1(w) \nsubseteq E_2(w)$. From the coNP-membership of the table containment problem, we can directly conclude the coNP-membership of the table equivalence problem.

In order to prove coNP-hardness, we reduce from $3$-CNF-Satisfiability. Let $F$ be a 3-CNF formula over variables $V = \{v_1, v_2, \ldots, v_n\}$. Let $\Sigma = \{\ta\}$ and let $\Gamma = \{t, f\}$ (hence, we have $|\Sigma| = 1$, $|\Gamma| = 2$). We note that any $\Sigma$-signed $\Gamma$-marker string $W$ of size $n$ that only contains markers $\{t\}_{\ta}$ and $\{f\}_{\ta}$ (and no occurrences of $\{t, f\}_{\ta}$ or $\emptyset_{\ta}$) can be interpreted as an assignment $\pi : V \to \{0, 1\}$, i.e., for every $i \in \{1, 2, \ldots, n\}$, we have $\pi(v_i) = 0$ if $W[i] = \{f\}_{\ta}$ and $\pi(v_i) = 1$ if $W[i] = \{t\}_{\ta}$. 

We construct in polynomial time a DFA $M_1$ that accepts all possible $\Sigma$-signed $\Gamma$-marker strings $W$ of size $n$ that do not contain occurrences of $\{t, f\}_{\ta}$ or $\emptyset_{\ta}$ (i.e., all $\Sigma$-signed $\Gamma$-marker string $W$ that represent an assignment $\pi : V \to \{0, 1\}$). Moreover, in polynomial time we can construct an NFA $M_2$ that accepts all $\Sigma$-signed $\Gamma$-marker strings of size $n$ that represent non-accepting assignments of $F$. Indeed, for every clause $c$ of $F$, $M_2$ has a branch that accepts all $\Sigma$-signed $\Gamma$-marker strings $W$ that represent an assignment that does not satisfy $c$. For example, if $c = \{v_4, \neg v_7, v_{12}\}$, then the corresponding branch of $M_2$ accepts all strings $R_1 \cdot R_2 \cdot \ldots \cdot R_n$, where $R_4 = R_{12} = \{f\}_{\ta}$, $R_7 = \{t\}_{\ta}$ and $R_i \in \{\{t\}_{\ta}, \{f\}_{\ta}\}$ for every $i \in \{1, 2, \ldots, n\} \setminus \{4, 7, 12\}$. Note that $M_2$ has polynomial size. 

We observe that $\llbracket L(M_1) \rrbracket_{\Sigma,\Gamma}(\ta^n) \subseteq \llbracket L(M_2) \rrbracket_{\Sigma,\Gamma}(\ta^n)$ if and only if $\llbracket L(M_1) \rrbracket_{\Sigma,\Gamma}(\ta^n) = \llbracket L(M_2) \rrbracket_{\Sigma,\Gamma}(\ta^n)$ if and only if $F$ is not satisfiable. Since both $L(M_1)$ and $L(M_2)$ are finite languages, the extractors $\llbracket L(M_1) \rrbracket_{\Sigma,\Gamma}$ and $\llbracket L(M_2) \rrbracket_{\Sigma,\Gamma}$ have finite support.

It remains to discuss the tractable cases mentioned in the statement of the theorem. 

We first consider the table containment problem. Let $E_1, E_2$ be regular extractors represented by an NFA $M_1$ and a DFA $M_2$, and let $w$ be a string. In time $O(|M_2|)$, we can construct a DFA $M'_2$ with $L(M'_2) = \overline{L(M_2)}$, which means that $\llbracket L(M'_2) \rrbracket_{\Sigma,\Gamma} = \neg E_2$ and therefore $\llbracket L(M'_2) \rrbracket_{\Sigma,\Gamma}(w) = \overline{E_2(w)}$. Hence, $E_1(w) \nsubseteq E_2(w)$ if and only if $E_1(w) \cap \llbracket L(M'_2) \rrbracket_{\Sigma,\Gamma}(w) \neq \emptyset$. Consequently, in order to decide whether $E_1(w) \nsubseteq E_2(w)$ it is sufficient to decide whether $E_1(w) \cap \llbracket L(M'_2) \rrbracket_{\Sigma,\Gamma}(w) \neq \emptyset$, which, according to Theorem~\ref{regTableDisjointmentTheorem}, can be done in time $O(|M_1| |M'_2| |w|) = O(|M_1| |M_2| |w|)$.

If both $M_1$ and $M_2$ are DFAs, then we can decide whether $E_1(w) \subseteq E_2(w)$ and $E_2(w) \subseteq E_1(w)$ in time $O(|M_1| |M_2| |w|)$, which means that we can solve the table equivalence problem in time $O(|M_1| |M_2| |w|)$.
\end{proof}

\subsubsection{Context-Free Extractors}

The table emptiness problem for context-free extractors can be solved as efficiently as the membership problem for CFGs.\footnote{The upper bound mentioned in Theorem~\ref{tableEmptinessContextFreeTheorem} can be improved by using Valiants parsing algorithm~\cite{Valiant1975}; we mention the CYK-based bound for simplicity.}

\begin{theorem}\label{tableEmptinessContextFreeTheorem}
The table emptiness problem for context-free extractors can be solved in time $O(|E| |w|^3)$.
\end{theorem}

\begin{proof}
Let $E$ be a context-free extractor represented by a context-free grammar $G$, and let $w$ be a string. Recall that $E(w) \neq \emptyset$ if and only if there is a marked string $W \in L(G)$ with $\sign(W) = w$. We obtain a context-free grammar $G'$ from $G$ by replacing each occurrence of a marker $X_{\ta}$ in any rule of $G$ by the marker $\emptyset_{\ta}$. Obviously, $G'$ can be constructed in time $O(G)$ and $|G'| = O(|G|)$. We observe that there is a marked string $W \in L(G)$ with $\sign(W) = w$ if and only if $\emptyset_{w[1]} \emptyset_{w[2]} \ldots \emptyset_{w[|w|]} \in L(G')$. Hence, we can decide whether $E(w) \neq \emptyset$ by checking $\emptyset_{w[1]} \emptyset_{w[2]} \ldots \emptyset_{w[|w|]} \in L(G')$, which can be done in time $O(|G'| |w|^3)$.
\end{proof}

We next observe that the table disjointness problem, which is tractable for regular extractors, becomes intractable for context-free extractors. 

\begin{theorem}\label{tableDisjointnessContextFreeTheorem}
The table disjointness problem for context-free extractors is NP-complete, even for extractors with finite support and if $|\Sigma| = 2$, but it can be solved in polynomial time, provided that one of the two extractors is regular.
\end{theorem}

\begin{proof}
The NP-membership can be easily seen: For context-free extractors $E_1$ and $E_2$ represented by context-free grammars $G_1$ and $G_2$, and a string $w$, we guess a marker string $W$ with $\sign(W) = w$ and then check whether $W \in L(G_1)$ and $W \in L(G_2)$, which is characteristic for $E_1(w) \cap E_2(w) \neq \emptyset$.

For showing NP-hardness, we reduce from the bounded post correspondence problem (see~\cite[Section A4.2]{GareyJohnson1979}), which is defined as follows. As input we get a list of the form $(u_1, v_1), (u_2, v_2), \ldots, (u_n, v_n)$ of pairs of strings $(u_i, v_i) \in \Lambda^* \times \Lambda^*$ for some alphabet $\Lambda$, and a number $\kappa \leq n$. The question is whether there is a sequence $i_1, i_2, \ldots, i_{q} \in \{1, 2, \ldots, n\}$ such that $q \leq \kappa$ and $u_{i_1} u_{i_2} \ldots u_{i_{q}} = v_{i_1} v_{i_2} \ldots v_{i_{q}}$. In the following, let us fix such an instance of the bounded post correspondence problem and, for convenience, we also define $p_{\max} = \max\{|u_i|, |v_i| \mid 1 \leq i \leq n\}$.

We will construct context-free grammars $G_1$ and $G_2$ that describe $\Sigma$-signed $\Gamma$-marker languages, where $\Sigma = \{\ta, \#\}$ and $\Gamma = \{1, 2, \ldots, n\} \cup \Lambda$; moreover, all $W \in L(G_1) \cup L(G_2)$ will be such that every marker set of a marker from $W$ is the empty set or a singleton, i.e., every symbol of $W$ has the form $\emptyset_{x}$ or $\{\gamma\}_{x}$ for some $\gamma \in \Gamma$ and $x \in \{\ta, \#\}$. Such $\Sigma$-signed $\Gamma$-marker strings can be interpreted as representing a string over $\Sigma$, i.e., $\sign(W) \in \Sigma^*$, and a string over $\Gamma$, which is obtained by replacing each marker $\{\gamma\}_{\ta}$ by $\gamma$ and each marker $\emptyset_{\ta}$ by $\varepsilon$. For convenience, we also use for an arbitrary string $u = \gamma_1 \gamma_2 \ldots \gamma_m $ with $\gamma_i \in \Gamma$ for $i \in \{1, 2, \ldots, m\}$ the notation $\{u\}_{\ta}$ as a short hand for the $\Sigma$-signed $\Gamma$-marker string $\{\gamma_1\}_{\ta} \{\gamma_2\}_{\ta} \ldots \{\gamma_m\}_{\ta}$ (note that $\sign(\{u\}_{\ta}) = \ta^{|u|}$ and the string over $\Gamma$ represented by $\{u\}_{\ta}$ is $u$). 

Let us now explain how $G_1$ is constructed. For every $q \in \{1, 2, \ldots, \kappa\}$ and $r \in \{1, 2, \ldots, \kappa \, p_{\max}\}$, we use a non-terminal $B_{q, r}$. For every $q \in \{1, 2, \ldots, \kappa-1\}$ and $r \in \{1, 2, \ldots, \kappa \, p_{\max}\}$ there is a rule 
\begin{equation*}
B_{q, r} \to \{i\}_{\ta} B_{q+1, r+|u_i|} \{u_i\}_{\ta}
\end{equation*}
and for every $q \in \{1, 2, \ldots, \kappa\}$ and $r \in \{1, 2, \ldots, \kappa \, p_{\max}\}$ there is a rule
\begin{equation*}
B_{q, r} \to (\emptyset_{\ta})^{\kappa-q} \emptyset_{\#} (\emptyset_{\ta})^{(\kappa \, p_{\max}) - r}\,.
\end{equation*}
Moreover, $S$ is the start non-terminal with a rule $S \to \{i\}_{\ta} B_{1, |u_i|} \{u_i\}_{\ta}$ for every $i \in \{1, 2, \ldots, n\}$. We observe that derivations of $G_1$ have the form
\begin{align*}
S &\to \{i_1\}_{\ta} B_{1, |u_{i_1}|} \{u_{i_1}\}_{\ta} \to \{i_1 i_2\}_{\ta} B_{2, |u_{i_1} u_{i_2}|} \{u_{i_2} u_{i_1}\}_{\ta}\\
&\to^* \{i_1 \ldots i_q\}_{\ta} B_{q, |u_{i_1} \ldots u_{i_q}|} \{u_{i_q} \ldots u_{i_1}\}_{\ta}\\
&\to \{i_1 \ldots i_q\}_{\ta} (\emptyset_{\ta})^{\kappa-q} \emptyset_{\#} (\emptyset_{\ta})^{(\kappa \, p_{\max}) - |u_{i_1} \ldots u_{i_q}|} \{u_{i_q} \ldots u_{i_1}\}_{\ta}
\end{align*}
for some $i_1, i_2, \ldots, i_{q} \in \{1, 2, \ldots, n\}$ and $q \leq \kappa$. Note that every $W \in L(G_1)$ satisfies $\sign(W) = \ta^{\kappa} \# \ta^{p_{\max} \, \kappa}$; thus, every $W \in L(G_1)$ describes a $\Gamma$-tuple $\llbracket W \rrbracket_{\Sigma, \Gamma}$ for the same string $w = \ta^{\kappa} \# \ta^{p_{\max} \, \kappa}$. In an analogous way, we can also construct a grammar $G_2$ that generates all marker strings 
\begin{equation*}
\{i_1 \ldots i_q\}_{\ta} (\emptyset_{\ta})^{\kappa-q} \emptyset_{\#} (\emptyset_{\ta})^{(\kappa \, p_{\max}) - |v_{i_1} \ldots v_{i_q}|} \{v_{i_q} \ldots v_{i_1}\}_{\ta}
\end{equation*}
for some $i_1, i_2, \ldots, i_{q} \in \{1, 2, \ldots, n\}$ and $q \leq \kappa$, which also represent $\Gamma$-tuples for $w$.

Let us observe that grammars $G_1$ and $G_2$ are of polynomial size and therefore can be constructed in polynomial time. The grammar $G_1$ has $\kappa^2 p_{\max}$ nonterminals and each nonterminal has at most $n + 1$ rules and each right-hand side of a rule has a size bounded by $2 \kappa p_{\max} + 1$. Since $\kappa \leq n$, the size of $G_1$ is polynomial in the size of the bounded PCP instance. The same applies to grammar $G_2$.

Now if $t \in \llbracket L(G_1) \rrbracket_{\Sigma, \Gamma}(w) \cap \llbracket L(G_2) \rrbracket_{\Sigma, \Gamma}(w)$, then the marker string $W$ with $\llbracket W \rrbracket_{\Sigma, \Gamma} = t$ satisfies $W \in L(G_1) \cap L(G_2)$. Thus, there are some $i_1, i_2, \ldots, i_{q}, j_1, j_2, \ldots, j_{q'} \in \{1, 2, \ldots, n\}$ and $q, q' \leq \kappa$ such that 
\begin{align*}
&\{i_1 \ldots i_q\}_{\ta} (\emptyset_{\ta})^{\kappa-q} \emptyset_{\#} (\emptyset_{\ta})^{(\kappa \, p_{\max}) - |u_{i_1} \ldots u_{i_q}|} \{u_{i_q} \ldots u_{i_1}\}_{\ta} =\\
&\{j_1 \ldots j_{q'}\}_{\ta} (\emptyset_{\ta})^{\kappa-{q'}} \emptyset_{\#} (\emptyset_{\ta})^{(\kappa \, p_{\max}) - |v_{j_1} \ldots v_{j_{q'}}|} \{v_{j_{q'}} \ldots v_{j_1}\}_{\ta}\,,
\end{align*}
which is only possible if $q = q'$, $(i_1, i_2, \ldots, i_{q}) = (j_1, j_2, \ldots, j_{q'})$ and $u_{i_1} \ldots u_{i_q} = v_{i_1} \ldots v_{i_q}$. 

Conversely, if there is some $i_1, i_2, \ldots, i_{q} \in \{1, 2, \ldots, n\}$ and $q \leq \kappa$ with $u_{i_1} \ldots u_{i_q} = v_{i_1} \ldots v_{i_q}$, then 
\begin{align*}
W_1 &= \{i_1 \ldots i_q\}_{\ta} (\emptyset_{\ta})^{\kappa-q} \emptyset_{\#} (\emptyset_{\ta})^{(\kappa \, p_{\max}) - |u_{i_1} \ldots u_{i_q}|} \{u_{i_q} \ldots u_{i_1}\}_{\ta} \in L(G_1)\,,\\
W_2 &= \{i_1 \ldots i_q\}_{\ta} (\emptyset_{\ta})^{\kappa-q} \emptyset_{\#} (\emptyset_{\ta})^{(\kappa \, p_{\max}) - |v_{i_1} \ldots v_{i_q}|} \{v_{i_q} \ldots v_{i_1}\}_{\ta} \in L(G_2)\,,
\end{align*}
and $W_1 = W_2$. This implies that $\llbracket W_1 \rrbracket_{\Sigma, \Gamma} \in \llbracket L(G_1) \rrbracket_{\Sigma, \Gamma}(w) \cap \llbracket L(G_2) \rrbracket_{\Sigma, \Gamma}(w)$. This concludes the proof of NP-hardness of the table disjointness problem.

It remains to show that the table disjointness problem for context-free extractors can be solved in polynomial time, provided that one of the two extractors is regular.

Let $E_1$ and $E_2$ be context-free extractors represented by a CFG $G_1$ and an NFA $M_2$, and let $w$ be a string. From $M_2$, we can easily obtain in polynomial time an NFA $M_{2, w}$ with $L(M_{2, w}) = \slice{w}(L(M_2))$ (see, e.\,g., proof of Theorem~\ref{regTableDisjointmentTheorem}). We observe that $E_1(w) \cap E_2(w) \neq \emptyset$ if and only if $L(G_1) \cap L(M_{2, w}) \neq \emptyset$. In order to decide the latter, we first construct a context-free grammar $G'$ with $L(G') = L(G_1) \cap L(M_{2, w})$, which can be done in polynomial time, and then we check whether $L(G') = \emptyset$, which is also possible in polynomial time.
\end{proof}

Finally, we consider the table containment and table equivalence problem for context-free extractors.

\begin{theorem}\label{tableContainmentEquivalenceContextFreeTheorem}
The table containment and equivalence problem for context-free extractors is coNP-complete, but the table containment problem for context-free extractors can be solved in polynomial time, provided that $E_2$ is given by a DFA.
\end{theorem}

\begin{proof}
The coNP-hardness of the table containment and table equivalence problem for context-free extractors is a direct consequence of Theorem~\ref{coNPHardnessRegularContainment}. The coNP-membership is easy to see: Let $M_1$ and $M_2$ be CFG that represent context-free extractors $E_1$ and $E_2$, respectively. Then we can guess a marker string $W$ with $\sign(W) = w$ and then check whether $W \in L(G_1)$ and $W \notin L(G_2)$, which is characteristic for $E_1(w) \nsubseteq E_2(w)$. The coNP-membership of the table equivalence problem follows directly from the coNP-membership of the table containment problem. 

It remains to show that the table containment problem for context-free extractors can be solved in polynomial time, provided that $E_2$ is given by a DFA.

Let $E_1$ and $E_2$ be context-free extractors represented by a CFG $G_1$ and an DFA $M_2$, and let $w$ be a string. We observe that $E_1(w) \nsubseteq E_2(w)$ if and only if $E_1(w) \cap \neg E_2(w) \neq \emptyset$. Since $M_2$ is a DFA, we can easily construct a DFA $M'_2$ with $L(M'_2) = \overline{L(M_2)}$ in polynomial time. Since $\llbracket L(M'_2) \rrbracket_{\Sigma, \Gamma} = \neg E_2$ (see Proposition~\ref{languageOperationsBooleanProposition}), we only have to check whether $E_1(w) \cap \neg E_2(w) \neq \emptyset$ for a context-free extractor $E_1$ and a regular extractor $\neg E_2$, which, according to Theorem~\ref{tableDisjointnessContextFreeTheorem}, can be done in polynomial time.
\end{proof}

\section{Conclusions}

This paper attempts to answer the following question: Given the fact that the work on information extraction in database theory is based on concepts and techniques of classical formal language theory, is there a unifying framework rooted purely in formal language theory (independent on specific data management tasks)? In particular, such a framework should be robust in the following sense: While existing information extraction techniques in database theory are somehow ``based on regular languages'', but ultimately designed in an ad-hoc way to solve a specific data management task at hand, our framework should allow to simply replace ``regular languages'' by just any language class. 

There are several aspects of this approach that might be beneficial for the future. Most importantly, we might identify new theoretical questions and worthwhile research tasks in formal language theory, which will allow us to progress this traditional field of theoretical research. For example, decision problems like membership, inclusion, equivalence, universality etc. are well-investigated for many language classes, but the table problems (see Section~\ref{sec:tableProbs}) are a new set of decision problems that do not arise from classical considerations in formal languages (intuitively speaking, table problems are problems concerned with ``finite slices'' of languages). Moreover, while undecidability is a very common obstacle for problems on formal languages, the table problems are all trivially decidable as long as the membership problem of the underlying class of marker languages is decidable (a property shared by virtually all useful language classes). This can be a new play area for complexity theoretical and algorithmic research within formal language theory.

Another observation is that regular extractors can also be seen as a restricted form of transducers (called \emph{annotation transducers} or \emph{annotation automata} in~\cite{MunozRiveros2025,GawrychowskiEtAl2024}), i.e., a transducer that can either erase an input symbol, or replace it by a marker set. However, in order to get the complete information of the corresponding $\Gamma$-tuple, we would also need along with the marker set the position of the input symbol in the overall input string (which, technically, requires an unbounded output alphabet). Hence, we are dealing with a setting that is close to, but not really covered by existing models in the theory of transducers. 

In summary, while techniques from formal languages can be beneficially exploited for data management tasks, one might also go in the other direction and ask whether these tools based on formal languages imply some interesting new theoretical questions to be investigated in the broader, more general setting of formal language theory. 

Finally, there might also be some practical implications. Arguably, extractor classes based on languages that are strictly more powerful than context-free languages likely suffer from intractability. However, there are many well-investigated subregular language classes as well as language classes sandwiched between regular and context-free languages. All of those are possible candidates for practically relevant extractor classes.

\section*{Acknowledgments}

The author is supported by the German Research Foundation (Deutsche Forschungsgemeinschaft, DFG) – project number 522576760 (gef\"{o}rdert durch die Deutsche Forschungsgemeinschaft (DFG) – Projektnummer 522576760).

 \bibliographystyle{plain}
 \bibliography{bibfile}

@inproceedings{Bagan2006,
  author    = {Guillaume Bagan},
  title     = {{MSO} Queries on Tree Decomposable Structures Are Computable with
               Linear Delay},
  booktitle = {Computer Science Logic, 20th International Workshop, {CSL} 2006, 15th
               Annual Conference of the EACSL, Szeged, Hungary, September 25-29,
               2006, Proceedings},
  pages     = {167--181},
  year      = {2006},
}

@inproceedings{FreydenbergerEtAl2018,
  author    = {Dominik D. Freydenberger and
               Benny Kimelfeld and
               Liat Peterfreund},
  title     = {Joining Extractions of Regular Expressions},
  booktitle = {Proceedings of the 37th {ACM} {SIGMOD-SIGACT-SIGAI} Symposium on Principles of Database Systems, Houston, TX, USA, June 10-15, 2018},
  pages     = {137--149},
  year      = {2018},
  doi       = {10.1145/3196959.3196967}
}

@article{FaginEtAl2015,
  author       = {Ronald Fagin and
                  Benny Kimelfeld and
                  Frederick Reiss and
                  Stijn Vansummeren},
  title        = {Document Spanners: {A} Formal Approach to Information Extraction},
  journal      = {J. {ACM}},
  volume       = {62},
  number       = {2},
  pages        = {12:1--12:51},
  year         = {2015},
  url          = {https://doi.org/10.1145/2699442},
  doi          = {10.1145/2699442},
  bibsource    = {dblp computer science bibliography, https://dblp.org}
}

@inproceedings{PeterfreundEtAl2019,
  author    = {Liat Peterfreund and
               Dominik D. Freydenberger and
               Benny Kimelfeld and
               Markus Kr{\"{o}}ll},
  title     = {Complexity Bounds for Relational Algebra over Document Spanners},
  booktitle = {Proceedings of the 38th {ACM} {SIGMOD-SIGACT-SIGAI} Symposium on Principles
               of Database Systems, {PODS} 2019, Amsterdam, The Netherlands, June
               30 - July 5, 2019.},
  pages     = {320--334},
  year      = {2019}
}

@article{AmarilliEtAl2021,
  author    = {Antoine Amarilli and
               Pierre Bourhis and
               Stefan Mengel and
               Matthias Niewerth},
  title     = {Constant-Delay Enumeration for Nondeterministic Document Spanners},
  journal   = {{ACM} Trans. Database Syst.},
  volume    = {46},
  number    = {1},
  pages     = {2:1--2:30},
  year      = {2021},
  doi       = {10.1145/3436487},
}

@inproceedings{MaturanaEtAl2018,
  author    = {Francisco Maturana and
               Cristian Riveros and
               Domagoj Vrgoc},
  title     = {Document Spanners for Extracting Incomplete Information: Expressiveness
               and Complexity},
  booktitle = {Proceedings of the 37th {ACM} {SIGMOD-SIGACT-SIGAI} Symposium on Principles
               of Database Systems, Houston, TX, USA, June 10-15, 2018},
  pages     = {125--136},
  year      = {2018},
  doi       = {10.1145/3196959.3196968}
}

@article{FlorenzanoEtAl2020,
  author    = {Fernando Florenzano and
               Cristian Riveros and
               Mart{\'{\i}}n Ugarte and
               Stijn Vansummeren and
               Domagoj Vrgoc},
  title     = {Efficient Enumeration Algorithms for Regular Document Spanners},
  journal   = {{ACM} Trans. Database Syst.},
  volume    = {45},
  number    = {1},
  pages     = {3:1--3:42},
  year      = {2020},
  doi       = {10.1145/3351451},
}

@inproceedings{PeterfreundEtAl2019_2,
  author    = {Liat Peterfreund and
               Balder ten Cate and
               Ronald Fagin and
               Benny Kimelfeld},
  title     = {Recursive Programs for Document Spanners},
  booktitle = {22nd International Conference on Database Theory, {ICDT} 2019, March
               26-28, 2019, Lisbon, Portugal},
  pages     = {13:1--13:18},
  year      = {2019}
}

@article{Freydenberger2019,
  author       = {Dominik D. Freydenberger},
  title        = {A Logic for Document Spanners},
  journal      = {Theory Comput. Syst.},
  volume       = {63},
  number       = {7},
  pages        = {1679--1754},
  year         = {2019},
  url          = {https://doi.org/10.1007/s00224-018-9874-1},
  doi          = {10.1007/S00224-018-9874-1}
}

@article{FreydenbergerHolldack2018,
  author       = {Dominik D. Freydenberger and
                  Mario Holldack},
  title        = {Document Spanners: From Expressive Power to Decision Problems},
  journal      = {Theory Comput. Syst.},
  volume       = {62},
  number       = {4},
  pages        = {854--898},
  year         = {2018},
  url          = {https://doi.org/10.1007/s00224-017-9770-0},
  doi          = {10.1007/S00224-017-9770-0},
  bibsource    = {dblp computer science bibliography, https://dblp.org}
}

@inproceedings{BackursIndyk2016,
  author    = {Arturs Backurs and
               Piotr Indyk},
  title     = {Which Regular Expression Patterns Are Hard to Match?},
  booktitle = {{IEEE} 57th Annual Symposium on Foundations of Computer Science, {FOCS}
               2016, 9-11 October 2016, Hyatt Regency, New Brunswick, New Jersey,
               {USA}},
  pages     = {457--466},
  year      = {2016},
  url       = {https://doi.org/10.1109/FOCS.2016.56}
}

@inproceedings{FreydenbergerThompson2020,
  author    = {Dominik D. Freydenberger and
               Sam M. Thompson},
  title     = {Dynamic Complexity of Document Spanners},
  booktitle = {23rd International Conference on Database Theory, {ICDT} 2020, March
               30-April 2, 2020, Copenhagen, Denmark},
  pages     = {11:1--11:21},
  year      = {2020},
  url       = {https://doi.org/10.4230/LIPIcs.ICDT.2020.11}
}

@inproceedings{FreydenbergerThompson2022,
  author    = {Dominik D. Freydenberger and
               Sam M. Thompson},
  title     = {Splitting Spanner Atoms: {A} Tool for Acyclic Core Spanners},
  booktitle = {25th International Conference on Database Theory, {ICDT} 2022, March
               29 to April 1, 2022, Edinburgh, {UK} (Virtual Conference)},
  pages     = {10:1--10:18},
  year      = {2022},
  doi       = {10.4230/LIPIcs.ICDT.2022.10}
}

@phdthesis{Peterfreund2019PhD,
title = "The Complexity of Relational Queries over Extractions from Text",
author = "Liat Peterfreund",
year = "2019",
language = "English",
publisher = "Computer science department, Technion",
number = "PHD-2019-10"
}

@article{DoleschalEtAl2020,
  author       = {Johannes Doleschal and
                  Benny Kimelfeld and
                  Wim Martens and
                  Liat Peterfreund},
  title        = {Weight Annotation in Information Extraction},
  journal      = {Log. Methods Comput. Sci.},
  volume       = {18},
  number       = {1},
  year         = {2022},
  url          = {https://doi.org/10.46298/lmcs-18(1:21)2022},
  doi          = {10.46298/LMCS-18(1:21)2022}
}

@article{Peterfreund2023,
  author       = {Liat Peterfreund},
  title        = {Enumerating grammar-based extractions},
  journal      = {Discret. Appl. Math.},
  volume       = {341},
  pages        = {372--392},
  year         = {2023},
  url          = {https://doi.org/10.1016/j.dam.2023.08.014},
  doi          = {10.1016/J.DAM.2023.08.014}
}

@inproceedings{SchmidSchweikardtPODS2021,
  author    = {Markus L. Schmid and
               Nicole Schweikardt},
  title     = {Spanner Evaluation over {SLP}-Compressed Documents},
  booktitle = {PODS'21: Proceedings of the 40th {ACM} {SIGMOD-SIGACT-SIGAI} Symposium
               on Principles of Database Systems, Virtual Event, China, June 20-25,
               2021},
  pages     = {153--165},
  year      = {2021},
  doi       = {10.1145/3452021.3458325}
}

@inproceedings{SchmidSchweikardtPODS2022,
  author       = {Markus L. Schmid and
                  Nicole Schweikardt},
  title        = {Query Evaluation over {SLP}-Represented Document Databases with Complex
                  Document Editing},
  booktitle    = {{PODS} '22: International Conference on Management of Data, Philadelphia,
                  PA, USA, June 12 - 17, 2022},
  pages        = {79--89},
  year         = {2022},
  doi          = {10.1145/3517804.3524158}
}

@book{HopcroftEtAlBook2007,
  author    = {John E. Hopcroft and
               Rajeev Motwani and
               Jeffrey D. Ullman},
  title     = {Introduction to automata theory, languages, and computation, 3rd Edition},
  series    = {Pearson international edition},
  publisher = {Addison-Wesley},
  year      = {2007},
  isbn      = {978-0-321-47617-3}
}

@article{DBLP:journals/sigmod/AmarilliBMN20,
  author    = {Antoine Amarilli and
               Pierre Bourhis and
               Stefan Mengel and
               Matthias Niewerth},
  title     = {Constant-Delay Enumeration for Nondeterministic Document Spanners},
  journal   = {{SIGMOD} Rec.},
  volume    = {49},
  number    = {1},
  pages     = {25--32},
  year      = {2020},
  doi       = {10.1145/3422648.3422655}
}

@inproceedings{SchmidSchweikardt2022,
  author       = {Markus L. Schmid and
                  Nicole Schweikardt},
  title        = {Document Spanners - {A} Brief Overview of Concepts, Results, and Recent
                  Developments},
  booktitle    = {{PODS} '22: International Conference on Management of Data, Philadelphia,
                  PA, USA, June 12 - 17, 2022},
  pages        = {139--150},
  year         = {2022},
  doi          = {10.1145/3517804.3526069}
}

@article{MunozRiveros2025,
  author       = {Martin Mu{\~{n}}oz and Cristian Riveros},
  title        = {Constant-delay enumeration for SLP-compressed documents},
  journal      = {Log. Methods Comput. Sci.},
  volume       = {21},
  number       = {1},
  year         = {2025},
  url          = {https://doi.org/10.46298/lmcs-21(1:17)2025},
  doi          = {10.46298/LMCS-21(1:17)2025},
}

@article{DoleschalEtAl2023,
  author       = {Johannes Doleschal and
                  Benny Kimelfeld and
                  Wim Martens},
  title        = {The Complexity of Aggregates over Extractions by Regular Expressions},
  journal      = {Log. Methods Comput. Sci.},
  volume       = {19},
  number       = {3},
  year         = {2023},
  url          = {https://doi.org/10.46298/lmcs-19(3:12)2023},
  doi          = {10.46298/LMCS-19(3:12)2023}
}

@inproceedings{AmarilliEtAl2022,
  author       = {Antoine Amarilli and
                  Louis Jachiet and
                  Martin Mu{\~{n}}oz and
                  Cristian Riveros},
  title        = {Efficient Enumeration for Annotated Grammars},
  booktitle    = {{PODS} '22: International Conference on Management of Data, Philadelphia,
                  PA, USA, June 12 - 17, 2022},
  pages        = {291--300},
  year         = {2022},
  doi          = {10.1145/3517804.3526232}
}

@inproceedings{Schmid2024,
  author       = {Markus L. Schmid},
  title        = {The Information Extraction Framework of Document Spanners - {A} Very Informal Survey},
  booktitle    = {{SOFSEM} 2024: Theory and Practice of Computer Science - 49th International
                  Conference on Current Trends in Theory and Practice of Computer Science,
                  {SOFSEM} 2024, Cochem, Germany, February 19-23, 2024, Proceedings},
  pages        = {3--22},
  year         = {2024},
  url          = {https://doi.org/10.1007/978-3-031-52113-3\_1},
  doi          = {10.1007/978-3-031-52113-3\_1}
}

@book{Libkin2004,
  author       = {Leonid Libkin},
  title        = {Elements of Finite Model Theory},
  series       = {Texts in Theoretical Computer Science. An {EATCS} Series},
  publisher    = {Springer},
  year         = {2004},
  url          = {http://www.cs.toronto.edu/\%7Elibkin/fmt},
  doi          = {10.1007/978-3-662-07003-1},
  isbn         = {3-540-21202-7}
}

@article{AbboudEtAl2018,
  author       = {Amir Abboud and
                  Arturs Backurs and
                  Virginia {Vassilevska Williams}},
  title        = {If the Current Clique Algorithms Are Optimal, so Is {V}aliant's Parser},
  journal      = {{SIAM} J. Comput.},
  volume       = {47},
  number       = {6},
  pages        = {2527--2555},
  year         = {2018},
  url          = {https://doi.org/10.1137/16M1061771},
  doi          = {10.1137/16M1061771},
}

@article{SchmidSchweikardt2024,
  author       = {Markus L. Schmid and
                  Nicole Schweikardt},
  title        = {Refl-Spanners: {A} Purely Regular Approach to Non-Regular Core Spanners},
  journal      = {Log. Methods Comput. Sci.},
  volume       = {20},
  number       = {4},
  year         = {2024},
  url          = {https://doi.org/10.46298/lmcs-20(4:16)2024},
  doi          = {10.46298/LMCS-20(4:16)2024},
}

@article{GawrychowskiEtAl2024,
  author       = {Pawel Gawrychowski and
                  Florin Manea and
                  Markus L. Schmid},
  title        = {Revisiting Weighted Information Extraction: {A} Simpler and Faster
                  Algorithm for Ranked Enumeration},
  journal      = {Proc. {ACM} Manag. Data},
  volume       = {2},
  number       = {5},
  pages        = {222:1--222:19},
  year         = {2024},
  url          = {https://doi.org/10.1145/3695840},
  doi          = {10.1145/3695840},
}

@book{GareyJohnson1979,
  author       = {Michael R. Garey and
                  David S. Johnson},
  title        = {Computers and Intractability: {A} Guide to the Theory of NP-Completeness},
  publisher    = {W. H. Freeman},
  year         = {1979},
  isbn         = {0-7167-1044-7}
}

@article{Valiant1975,
  author       = {Leslie G. Valiant},
  title        = {General Context-Free Recognition in Less than Cubic Time},
  journal      = {J. Comput. Syst. Sci.},
  volume       = {10},
  number       = {2},
  pages        = {308--315},
  year         = {1975},
  url          = {https://doi.org/10.1016/S0022-0000(75)80046-8},
  doi          = {10.1016/S0022-0000(75)80046-8}
}

@article{LohreySchmid2026,
  author       = {Markus Lohrey and
                  Markus L. Schmid},
  title        = {{MSO}-Enumeration Over {SLP}-Compressed Unranked Forests},
  journal      = {TheoretiCS},
  volume       = {5},
  year         = {2026},
  url          = {https://doi.org/10.46298/theoretics.26.6},
  doi          = {10.46298/THEORETICS.26.6}
}

@book{HopcroftEtAl,
  author       = {John E. Hopcroft and
                  Rajeev Motwani and
                  Jeffrey D. Ullman},
  title        = {Introduction to automata theory, languages, and computation, 3rd Edition},
  series       = {Pearson international edition},
  publisher    = {Addison-Wesley},
  year         = {2007},
  isbn         = {978-0-321-47617-3}
}

@inproceedings{StockmeyerMeyer1973,
  author       = {Larry J. Stockmeyer and
                  Albert R. Meyer},
  editor       = {Alfred V. Aho and
                  Allan Borodin and
                  Robert L. Constable and
                  Robert W. Floyd and
                  Michael A. Harrison and
                  Richard M. Karp and
                  H. Raymond Strong},
  title        = {Word Problems Requiring Exponential Time: Preliminary Report},
  booktitle    = {Proceedings of the 5th Annual {ACM} Symposium on Theory of Computing,
                  April 30 - May 2, 1973, Austin, Texas, {USA}},
  pages        = {1--9},
  publisher    = {{ACM}},
  year         = {1973},
  url          = {https://doi.org/10.1145/800125.804029},
  doi          = {10.1145/800125.804029}
}

\appendix

\section{Full Proofs for Section~\ref{sec:operations}}\label{sec:operationsAppendix}

\begin{proposition}
For $i \in \{1, 2, 3\}$, let $t_i$ be a $\Gamma_i$-tuple for $w_i$, let $T_i$ be a $\Gamma_i$-table for $w_i$, and let $E_i$ be a $\Gamma_i$-extractor. Then we have $(t_1 \cdot t_2) \cdot t_3 = t_1 \cdot (t_2 \cdot t_3)$, $(T_1 \cdot T_2) \cdot T_3 = T_1 \cdot (T_2 \cdot T_3)$ and $(E_1 \cdot E_2) \cdot E_3 = E_1 \cdot (E_2 \cdot E_3)$.
\end{proposition}

\begin{proof}
The first point follows directly from the definition of the concatenation operation.

The second point can be shown as follows (observe that this assumes the associativity for tuples):
\begin{align*}
(T_1 \cdot T_2) \cdot T_3 &= \{t_1 \cdot t_2 \mid t_1 \in T_1, t_2 \in T_2\} \cdot T_3\\
& = \{(t_1 \cdot t_2) \cdot t_3 \mid t_1 \in T_1, t_2 \in T_2, t_3 \in T_3\}\\
& = \{t_1 \cdot (t_2 \cdot t_3) \mid t_1 \in T_1, t_2 \in T_2, t_3 \in T_3\}\\
& = T_1 \cdot \{t_2 \cdot t_3 \mid t_2 \in T_2, t_3 \in T_3\}\\
& = T_1 \cdot (T_2 \cdot T_3)
\end{align*}

For the third point, let $w \in \Sigma^*$ be arbitrarily chosen. Then we have (observe that this assumes the associativity for tables):
\begin{align*}
((E_1 \cdot E_2) \cdot E_3)(w) &= \bigcup_{w = u_1 \cdot u_2} (E_1 \cdot E_2)(u_1) \cdot E_3(u_2)\\
&= \bigcup_{w = u_1 \cdot u_2} (\bigcup_{u_1 = v_1 \cdot v_2} E_1(v_1) \cdot E_2(v_2)) \cdot E_3(u_2)\\
&= \bigcup_{w = v_1 \cdot v_2 \cdot u_2} (E_1(v_1) \cdot E_2(v_2)) \cdot E_3(u_2)\\
&= \bigcup_{w = v_1 \cdot v_2 \cdot u_2} E_1(v_1) \cdot (E_2(v_2) \cdot E_3(u_2))\\
&= \bigcup_{w = v_1 \cdot u'} E_1(v_1) \cdot (\bigcup_{u' = v_2 \cdot u_2} (E_2(v_2) \cdot E_3(u_2)))\\
&= \bigcup_{w = v_1 \cdot u'} E_1(v_1) \cdot (E_2 \cdot E_3)(u')\\
&= (E_1 \cdot (E_2 \cdot E_3))(w)
\end{align*}
\end{proof}

\section{Full Proofs for Section~\ref{sec:isomorphism}}\label{sec:IsomorphismAppendix}

The theorems of Section~\ref{sec:isomorphism} directly follow from the propositions in this section. Let us first take care of the Boolean set operations.

\begin{proposition}\label{languageOperationsBooleanProposition}
Let $L_1$ and $L_2$ be $\Sigma$-signed $\Gamma$-marker languages. Then $L_1 \odot L_2$ is a $\Sigma$-signed $\Gamma$-marker language with $\llbracket L_1 \odot L_2 \rrbracket_{\Sigma, \Gamma} = \llbracket L_1 \rrbracket_{\Sigma, \Gamma} \odot \llbracket L_2 \rrbracket_{\Sigma, \Gamma}$ for every $\odot \in \{\cup, \cap, \setminus\}$. Moreover, $\overline{L_1}$ is a $\Sigma$-signed $\Gamma$-marker language with $\llbracket \overline{L_1} \rrbracket_{\Sigma, \Gamma} = \neg \llbracket L_1 \rrbracket_{\Sigma, \Gamma}$.
\end{proposition}

\begin{proof}
We first prove the case $\odot = \cup$. We note that $L_1 \cup L_2$ is obviously a $\Sigma$-signed $\Gamma$-marker language. In order to prove that $\llbracket L_1 \cup L_2 \rrbracket_{\Sigma, \Gamma} = \llbracket L_1 \rrbracket_{\Sigma, \Gamma} \cup \llbracket L_2 \rrbracket_{\Sigma, \Gamma}$, let $w \in \Sigma^*$. 
\begin{align*}
&\llbracket L_1 \cup L_2 \rrbracket_{\Sigma, \Gamma}(w) = \{\llbracket W \rrbracket_{\Sigma, \Gamma} \mid W \in \slice{w}(L_1 \cup L_2)\} = \{\llbracket W \rrbracket_{\Sigma, \Gamma} \mid W \in \slice{w}(L_1) \cup \slice{w}(L_2)\} = \\
&\{\llbracket W \rrbracket_{\Sigma, \Gamma} \mid W \in \slice{w}(L_1)\} \cup \{\llbracket W \rrbracket_{\Sigma, \Gamma} \mid W \in \slice{w}(L_2)\} = \llbracket L_1 \rrbracket_{\Sigma, \Gamma}(w) \cup \llbracket L_2 \rrbracket_{\Sigma, \Gamma}(w)
\end{align*}
It can be easily verified that the above equation is also correct for the other choices of $\odot \in \{\cup, \cap, \setminus\}$.

Note that $\overline{L_1}$ is a $\Sigma$-signed $\Gamma$-marker language, and recall that $\neg \llbracket L_1 \rrbracket_{\Sigma, \Gamma}$ is a $(\Sigma, \Gamma)$-extractor defined by $(\neg \llbracket L_1 \rrbracket_{\Sigma, \Gamma})(w) = \overline{\llbracket L_1 \rrbracket_{\Sigma, \Gamma}(w)} = \{\llbracket W \rrbracket_{\Sigma, \Gamma} \mid W \notin L_1 \wedge \sign(W) = w\}$, for every $w \in \Sigma^*$. Now let $w \in \Sigma^*$ be arbitrarily chosen.

\begin{align*}
&\llbracket \overline{L_1} \rrbracket_{\Sigma, \Gamma}(w) = \{\llbracket W \rrbracket_{\Sigma, \Gamma} \mid W \in \slice{w}(\overline{L_1})\} = \\
&\{\llbracket W \rrbracket_{\Sigma, \Gamma} \mid W \in \overline{L_1} \wedge \sign(W) = w\} = \\
&\{\llbracket W \rrbracket_{\Sigma, \Gamma} \mid W \notin L_1 \wedge \sign(W) = w\} = \\
&\overline{\llbracket L_1 \rrbracket_{\Sigma, \Gamma}(w)} = 
(\neg \llbracket L_1 \rrbracket_{\Sigma, \Gamma})(w)\,.
\end{align*}

\end{proof}

The next proposition states that the concatenation and the join variants on marker strings correspond to the respective operations on tuples. 

\begin{proposition}\label{binaryLanguageOperationsStringsProposition}
Let $W_1$ be a $\Sigma$-signed $\Gamma_1$-marker string and let $W_2$ be a $\Sigma$-signed $\Gamma_2$-marker string. Then $W_1 \cdot W_2$ is a $\Sigma$-signed $(\Gamma_1 \cup \Gamma_2)$-marker string with $\llbracket W_1 \cdot W_2 \rrbracket_{\Sigma, \Gamma_1 \cup \Gamma_2} = \llbracket W_1 \rrbracket_{\Sigma, \Gamma_1} \cdot \llbracket W_2 \rrbracket_{\Sigma, \Gamma_2}$. If further $\sign(W_1) = \sign(W_2)$, then $W_1 \circ W_2$ is a $\Sigma$-signed $(\Gamma_1 \cup \Gamma_2)$-marker string with $\llbracket W_1 \circ W_2 \rrbracket_{\Sigma, \Gamma_1 \cup \Gamma_2} = \llbracket W_1 \rrbracket_{\Sigma, \Gamma_1} \circ \llbracket W_2 \rrbracket_{\Sigma, \Gamma_2}$ for every $\circ \in \{\bowtie, \bowtie_{\cup}, \bowtie_{\cap}, \bowtie_{\setminus}\}$.
\end{proposition}

\begin{proof}
By definition, $W_1 \cdot W_2$ is a  $\Sigma$-signed $(\Gamma_1 \cup \Gamma_2)$-marker string for $\sign(W_1) \cdot \sign(W_2)$. We next prove that $\llbracket W_1 \cdot W_2 \rrbracket_{\Sigma, \Gamma_1 \cup \Gamma_2} = \llbracket W_1 \rrbracket_{\Sigma, \Gamma_1} \cdot \llbracket W_2 \rrbracket_{\Sigma, \Gamma_2}$. Let $m_1 = |W_1|$ and $m_2 = |W_2|$. For every $x \in \Gamma_1 \cup \Gamma_2$ and for every $i \in \{1, 2, \ldots, m_1 + m_2\}$, we have that
\begin{align*}
&i \in \llbracket W_1 \cdot W_2 \rrbracket_{\Sigma, \Gamma_1 \cup \Gamma_2}(x) \Leftrightarrow x \in (W_1 \cdot W_2)[i] \Leftrightarrow\\
&(1 \leq i \leq m_1 \wedge x \in W_1[i]) \vee (m_1 + 1 \leq i \leq m_1 + m_2 \wedge x \in W_2[i-m_1] \Leftrightarrow\\
&(1 \leq i \leq m_1 \wedge i \in \llbracket W_1 \rrbracket_{\Sigma, \Gamma_1}(x)) \vee (m_1 + 1 \leq i \leq m_1 + m_2 \wedge (i-m_1) \in \llbracket W_2 \rrbracket_{\Sigma, \Gamma_2}(x) \Leftrightarrow\\
&i \in (\llbracket W_1 \rrbracket_{\Sigma, \Gamma_1} \cdot \llbracket W_2 \rrbracket_{\Sigma, \Gamma_2})(x)\,.
\end{align*}
Thus, $\llbracket W_1 \cdot W_2 \rrbracket_{\Sigma, \Gamma_1 \cup \Gamma_2} = \llbracket W_1 \rrbracket_{\Sigma, \Gamma_1} \cdot \llbracket W_2 \rrbracket_{\Sigma, \Gamma_2}$.

Next, let us assume that $\sign(W_1) = \sign(W_2)$. By definition, $W_1 \bowtie_{\cup} W_2$ is a $\Sigma$-signed $(\Gamma_1 \cup \Gamma_2)$-marker string. We next show that $\llbracket W_1 \bowtie_{\cup} W_2 \rrbracket_{\Sigma, \Gamma_1 \cup \Gamma_2} = \llbracket W_1 \rrbracket_{\Sigma, \Gamma_1} \bowtie_{\cup} \llbracket W_2 \rrbracket_{\Sigma, \Gamma_2}$ holds. For every $x \in \Gamma_1 \cup \Gamma_2$ and every $i \in \{1, 2, \ldots, |W_1|\}$, we have that
\begin{align*}
&i \in \llbracket W_1 \bowtie_{\cup} W_2 \rrbracket_{\Sigma, \Gamma_1 \cup \Gamma_2}(x) \Leftrightarrow x \in (W_1 \bowtie_{\cup} W_2)[i] \Leftrightarrow x \in W_1[i] \vee x \in W_2[i] \Leftrightarrow \\
&i \in \llbracket W_1 \rrbracket_{\Sigma, \Gamma_1}(x) \vee i \in \llbracket W_2 \rrbracket_{\Sigma, \Gamma_2}(x) \Leftrightarrow i \in (\llbracket W_1 \rrbracket_{\Sigma, \Gamma_1} \bowtie_{\cup} \llbracket W_2 \rrbracket_{\Sigma, \Gamma_2})(x)\,.
\end{align*}
The cases $\circ \in \{\bowtie, \bowtie_{\cap}, \bowtie_{\setminus}\}$ follow analogously. 
\end{proof}

With the help of Proposition~\ref{binaryLanguageOperationsStringsProposition}, we can conclude that the concatenation, the join variants and Kleene-star on marker languages correspond to the respective operations on extractors.

\begin{proposition}\label{binaryLanguageOperationsLanguagesProposition}
Let $L_1$ and $L_2$ be $\Sigma$-signed $\Gamma_1$- and $\Sigma$-signed $\Gamma_2$-marker languages, respectively, and let $\circ \in \{\bowtie, \bowtie_{\cup}, \bowtie_{\cap}, \bowtie_{\setminus}, \cdot\}$. Then $L_1 \circ L_2$ is a $\Sigma$-signed $(\Gamma_1 \cup \Gamma_2)$-marker language with $\llbracket L_1 \circ L_2 \rrbracket_{\Sigma, \Gamma_1 \cup \Gamma_2} = \llbracket L_1 \rrbracket_{\Sigma, \Gamma_1} \circ \llbracket L_2 \rrbracket_{\Sigma, \Gamma_2}$, and $L_1^*$ is a $\Gamma_1$-marker language with $\llbracket L_1^* \rrbracket_{\Sigma, \Gamma_1} = (\llbracket L_1 \rrbracket_{\Sigma, \Gamma_2})^*$. 
\end{proposition}

\begin{proof}
We first show the statement for every $\circ \in \{\bowtie, \bowtie_{\cup}, \bowtie_{\cap}, \bowtie_{\setminus}\}$ (i.e., the concatenation and Kleene star is handled later). That $L_1 \circ L_2$ is a $\Sigma$-signed $(\Gamma_1 \cup \Gamma_2)$-marker language follows directly from the definition of the operation $\circ$. Recall that $L_1 \circ L_2 = \bigcup_{w \in \Sigma^*} \{W_1 \circ W_2 \mid W_1 \in \slice{w}(L_1), W_2 \in \slice{w}(L_2)\}$. In order to show $\llbracket L_1 \circ L_2 \rrbracket_{\Sigma, \Gamma_1 \cup \Gamma_2} = \llbracket L_1 \rrbracket_{\Sigma, \Gamma_1} \circ \llbracket L_2 \rrbracket_{\Sigma, \Gamma_2}$, let $w \in \Sigma^*$ be arbitrarily chosen. We have the following (note that we use Proposition~\ref{binaryLanguageOperationsStringsProposition}, i.e.,  $\llbracket W_1 \circ W_2 \rrbracket_{\Sigma, \Gamma_1 \cup \Gamma_2} = \llbracket W_1 \rrbracket_{\Sigma, \Gamma_1} \circ \llbracket W_2 \rrbracket_{\Sigma, \Gamma_2}$ for marker-strings $W_1$ and $W_2$):
\begin{align*}
&\llbracket L_1 \circ L_2 \rrbracket_{\Sigma, \Gamma_1 \cup \Gamma_2}(w) = \{\llbracket W \rrbracket_{\Sigma, \Gamma_1 \cup \Gamma_2} \mid W \in \slice{w}(L_1 \circ L_2)\} =\\
&\{\llbracket W \rrbracket_{\Sigma, \Gamma_1 \cup \Gamma_2} \mid W \in \slice{w}(\bigcup_{u \in \Sigma^*} \{W_1 \circ W_2 \mid W_1 \in \slice{u}(L_1), W_2 \in \slice{u}(L_2)\})\} =\\
&\{\llbracket W \rrbracket_{\Sigma, \Gamma_1 \cup \Gamma_2} \mid W \in \{W_1 \circ W_2 \mid W_1 \in \slice{w}(L_1), W_2 \in \slice{w}(L_2)\}\} =\\
& \{\llbracket W_1 \circ W_2 \rrbracket_{\Sigma, \Gamma_1 \cup \Gamma_2} \mid W_1 \in \slice{w}(L_1), W_2 \in \slice{w}(L_2)\}  = \\
& \{\llbracket W_1 \rrbracket_{\Sigma, \Gamma_1} \circ \llbracket W_2 \rrbracket_{\Sigma, \Gamma_2} \mid W_1 \in \slice{w}(L_1), W_2 \in \slice{w}(L_2)\} = \\
& \{\llbracket W_1 \rrbracket_{\Sigma, \Gamma_1} \mid W_1 \in \slice{w}(L_1)\} \circ \{\llbracket W_2 \rrbracket_{\Sigma, \Gamma_2} \mid W_2 \in \slice{w}(L_2)\} =\\
& \llbracket L_1\rrbracket_{\Sigma, \Gamma_1}(w) \circ \llbracket L_2 \rrbracket_{\Sigma, \Gamma_2}(w) = 
(\llbracket L_1\rrbracket_{\Sigma, \Gamma_1} \circ \llbracket L_2 \rrbracket_{\Sigma, \Gamma_2})(w)
\end{align*}
We next take care of the concatenation, and we first observe that $L_1 \cdot L_2$ is a $\Sigma$-signed $(\Gamma_1 \cup \Gamma_2)$-marker language by definition of the concatenation. Recall that $L_1 \cdot L_2 = \{W_1 \cdot W_2 \mid W_1 \in L_1, W_2 \in L_2\}$. In order to show $\llbracket L_1 \cdot L_2 \rrbracket_{\Sigma, \Gamma_1 \cup \Gamma_2} = \llbracket L_1 \rrbracket_{\Sigma, \Gamma_1} \cdot \llbracket L_2 \rrbracket_{\Sigma, \Gamma_2}$, let $w \in \Sigma^*$ be arbitrarily chosen. We have the following (note that we will use Proposition~\ref{binaryLanguageOperationsStringsProposition}, i.e.,  $\llbracket W_1 \cdot W_2 \rrbracket_{\Sigma, \Gamma_1 \cup \Gamma_2} = \llbracket W_1 \rrbracket_{\Sigma, \Gamma_1} \cdot \llbracket W_2 \rrbracket_{\Sigma, \Gamma_2}$ for a $\Sigma$-signed $\Gamma_1$-marker string $W_1$ and a $\Sigma$-signed $\Gamma_2$-marker string $W_2$; moreover, recall that $(E_1 \cdot E_2)(w) = \bigcup_{w = w_1 \cdot w_2} E_1(w_1) \cdot E_2(w_2)$ for a $(\Sigma, \Gamma_1)$-extractor $E_1$ and a $(\Sigma, \Gamma_2)$-extractor $E_2$):
\begin{align*}
&\llbracket L_1 \cdot L_2 \rrbracket_{\Sigma, \Gamma_1 \cup \Gamma_2}(w) = \{\llbracket W \rrbracket_{\Sigma, \Gamma_1 \cup \Gamma_2} \mid W \in \slice{w}(L_1 \cdot L_2)\} =\\
&\{\llbracket W \rrbracket_{\Sigma, \Gamma_1 \cup \Gamma_2} \mid W \in \slice{w}(\{W_1 \cdot W_2 \mid W_1 \in L_1, W_2 \in L_2\})\} =\\
&\{\llbracket W_1 \cdot W_2 \rrbracket_{\Sigma, \Gamma_1 \cup \Gamma_2} \mid W_1 \in L_1, W_2 \in L_2, \sign(W_1 \cdot W_2) = w\} =\\
&\bigcup_{w = w_1 w_2} \{\llbracket W_1 \cdot W_2 \rrbracket_{\Sigma, \Gamma_1 \cup \Gamma_2} \mid W_1 \in \slice{w_1}(L_1), W_2 \in \slice{w_2}(L_2)\} =\\
&\bigcup_{w = w_1 w_2} \{\llbracket W_1 \rrbracket_{\Sigma, \Gamma_1} \cdot \llbracket W_2 \rrbracket_{\Sigma, \Gamma_2} \mid W_1 \in \slice{w_1}(L_1), W_2 \in \slice{w_2}(L_2)\} =\\
& \bigcup_{w = w_1 w_2} \{\llbracket W_1 \rrbracket_{\Sigma, \Gamma_1} \mid W_1 \in \slice{w_1}(L_1)\} \cdot \{\llbracket W_2 \rrbracket_{\Sigma, \Gamma_2} \mid W_2 \in \slice{w_2}(L_2)\} =\\
& \bigcup_{w = w_1 w_2} \llbracket L_1 \rrbracket_{\Sigma, \Gamma_1}(w_1) \cdot \llbracket L_2 \rrbracket_{\Sigma, \Gamma_2}(w_2) = (\llbracket L_1 \rrbracket_{\Sigma, \Gamma_1} \cdot \llbracket L_2 \rrbracket_{\Sigma, \Gamma_2})(w)\,.
\end{align*}

It remains to take care of the Kleene star. To this end, let us first recall that $L_1^* = \bigcup_{k \geq 0} L_1^k$ for a $\Sigma$-signed $\Gamma_1$-marker language $L_1$, that $(E_1 \cdot E_2)(w) = \bigcup_{w = w_1 \cdot w_2} E_1(w_1) \cdot E_2(w_2)$ for a $(\Sigma, \Gamma_1)$-extractor $E_1$ and a $(\Sigma, \Gamma_2)$-extractor $E_2$, and that $(E_1)^* = \bigcup_{k \geq 1} (E_1)^k$ for a $(\Sigma, \Gamma_1)$-extractor $E$. Let $w \in \Sigma^*$ be arbitrarily chosen.

\begin{align*}
&\llbracket L_1^* \rrbracket_{\Sigma, \Gamma_1}(w) = \{\llbracket W \rrbracket_{\Sigma, \Gamma_1} \mid W \in \slice{w}(L_1^*)\} = \bigcup_{k \geq 1} \{\llbracket W \rrbracket_{\Sigma, \Gamma_1} \mid W \in \slice{w}(L_1^k)\} = \\
&\bigcup_{k \geq 1} \bigcup_{w = w_1 \ldots w_k} \{\llbracket W_1 \cdot \ldots \cdot W_k \rrbracket_{\Sigma, \Gamma_1} \mid W_i \in \slice{w_i}(L_1), i \in \{1, 2, \ldots, k\}\}= \\
&\bigcup_{k \geq 1} \bigcup_{w = w_1 \ldots w_k} \{\llbracket W_1 \rrbracket_{\Sigma, \Gamma_1} \cdot \ldots \cdot \llbracket W_k \rrbracket_{\Sigma, \Gamma_1} \mid W_i \in \slice{w_i}(L_1), i \in \{1, 2, \ldots, k\}\} = \\
&\bigcup_{k \geq 1} \bigcup_{w = w_1 \ldots w_k} \{\llbracket W_1 \rrbracket_{\Sigma, \Gamma_1} \mid W_1 \in \slice{w_1}(L_1)\} \cdot \ldots \cdot \{\llbracket W_k \rrbracket_{\Sigma, \Gamma_1} \mid W_k \in \slice{w_k}(L_1)\} =\\
&\bigcup_{k \geq 1} \bigcup_{w = w_1 \ldots w_k} \llbracket L_1 \rrbracket_{\Sigma, \Gamma_1}(w_1) \cdot \ldots \cdot \llbracket L_1 \rrbracket_{\Sigma, \Gamma_1}(w_k) = \\
& \bigcup_{k \geq 1} (\llbracket L_1 \rrbracket_{\Sigma, \Gamma_1}^k)(w) = (\bigcup_{k \geq 1} \llbracket L_1 \rrbracket_{\Sigma, \Gamma_1}^k)(w) = ((\llbracket L_1 \rrbracket_{\Sigma, \Gamma_1})^*)(w)
\end{align*}

\end{proof}

The next proposition establishes that the unary operations on marker strings correspond to the respective operations on tuples.

\begin{proposition}\label{unaryLanguageOperationsStringsProposition}
Let $W$ be a $\Sigma$-signed $\Gamma$-marker string. Then 
\begin{itemize}
\item $\llbracket \pi_{\Gamma'}(W) \rrbracket_{\Sigma, \Gamma'} = \pi_{\Gamma'}(\llbracket W \rrbracket_{\Sigma, \Gamma})$ for every $\Gamma' \subseteq \Gamma$. 
\item $\llbracket \colmerge_{y, y', \odot}(W) \rrbracket_{\Sigma, \Gamma \setminus \{y'\}} = \colmerge_{y, y', \odot}(\llbracket W \rrbracket_{\Sigma, \Gamma})$ for every $y, y' \in \Gamma$ and $\odot \in \{\cup, \cap, \setminus\}$.
\item $\llbracket \colrename_{y \to z}(W) \rrbracket_{\Sigma, (\Gamma \setminus \{y\}) \cup \{z\}} = \colrename_{y \to z}(\llbracket W \rrbracket_{\Sigma, \Gamma})$ for every $y \in \Gamma$ and $z \notin \Gamma$.
\end{itemize}
\end{proposition}

\begin{proof}
This follows directly from the definitions of these operators on tuples and marker strings. 
\end{proof}

Finally, we observe that these unary operations on marker languages correspond to the respective operations on extractors:

\begin{proposition}\label{unaryLanguageOperationsLanguagesProposition}
Let  $L$ be a $\Sigma$-signed $\Gamma$-marker language, let $\Gamma' \subseteq \Gamma$, let $y, y' \in \Gamma$, and let $z \notin \Gamma$.
Then $\pi_{\Gamma'}(L)$ is a $\Sigma$-signed $\Gamma'$-marker language with $\llbracket \pi_{\Gamma'}(L) \rrbracket_{\Sigma, \Gamma'} = \pi_{\Gamma'}(\llbracket L \rrbracket_{\Sigma, \Gamma})$, $f(L)$ is a $\Sigma$-signed $(\Gamma \setminus \{y'\})$-marker language with $\llbracket f(L) \rrbracket_{\Sigma, \Gamma \setminus \{y'\}} = f(\llbracket L \rrbracket_{\Sigma, \Gamma})$ for every $f \in \{\colmerge_{y, y', \cup}, \colmerge_{y, y', \cap}, \colmerge_{y, y', \setminus}\}$, and $\colrename_{y \to z}(L)$ is a $\Sigma$-signed $((\Gamma \setminus \{y\}) \cup \{z\})$-marker language with $\llbracket \colrename_{y \to z}(L) \rrbracket_{\Sigma, (\Gamma \setminus \{y\}) \cup \{z\}} = \colrename_{y \to z}(\llbracket L \rrbracket_{\Sigma, \Gamma})$.
\end{proposition}

\begin{proof}
Let $w \in \Sigma^*$ and $f \in \{\pi_{\Gamma'}, \colmerge_{y, y', \cup}, \colmerge_{y, y', \cap}, \colmerge_{y, y', \setminus}, \colrename_{y \to z}\}$. In the following, we will use Proposition~\ref{unaryLanguageOperationsStringsProposition}, i.e., $\llbracket f(W) \rrbracket_{\Sigma, \widehat{\Gamma}} = f(\llbracket W \rrbracket_{\Sigma, \Gamma})$ for $\Sigma$-signed $\widehat{\Gamma}$-marker strings $W$ (for the apropriate $\widehat{\Gamma}$, which depends on $f$). Also recall that $f(L) = \{f(W) \mid W \in L\}$ for a marker language $L$. We have that
\begin{align*}
&\llbracket f(L) \rrbracket_{\Sigma, \widehat{\Gamma}}(w) = \llbracket \{f(W) \mid W \in L\} \rrbracket_{\Sigma, \widehat{\Gamma}}(w) = \llbracket \{f(W) \mid W \in \slice{w}(L)\} \rrbracket_{\Sigma, \widehat{\Gamma}} = \\
&\{\llbracket f(W) \rrbracket_{\Sigma, \widehat{\Gamma}} \mid W \in \slice{w}(L)\} = \{f(\llbracket W \rrbracket_{\Sigma, \Gamma})  \mid W \in \slice{w}(L)\} =\\
&f(\{\llbracket W \rrbracket_{\Sigma, \Gamma}  \mid W \in \slice{w}(L)\}) = f(\llbracket L \rrbracket_{\Sigma, \Gamma}(w)) = (f(\llbracket L \rrbracket_{\Sigma, \Gamma}))(w)\,,
\end{align*}
where $\widehat{\Gamma}$ depends on $f$, i.e., $\widehat{\Gamma} = \Gamma'$ for $f = \pi_{\Gamma'}$,  $\widehat{\Gamma} = \Gamma \setminus \{y'\}$ for $f \in \{\colmerge_{y, y', \cup}, \colmerge_{y, y', \cap}, \colmerge_{y, y', \setminus}\}$, and $\widehat{\Gamma} = (\Gamma \setminus \{y\}) \cup \{z\}$ for $f = \colrename_{y \to z}$.
\end{proof}

\end{document}